\documentclass{article}
\pdfoutput=1
\usepackage[utf8]{inputenc}
\usepackage[english]{babel}
\usepackage{amsthm,mathtools}
\usepackage{amsthm}
\usepackage{geometry}
\usepackage{amsfonts}
\usepackage{color,soul}
\usepackage{amssymb}
\usepackage{tikz}
\usepackage{footmisc}
\usetikzlibrary{automata,arrows,positioning,calc}
\usepackage[T1]{fontenc}
\usepackage{lmodern}
\usepackage{hyperref}
\usepackage[export]{adjustbox}
\usepackage{subcaption}
\usepackage{graphicx}
\graphicspath{ {images/} }
\newtheorem{theorem}{Theorem}[section]
\newtheorem{lemma}[theorem]{Lemma}
\newtheorem{fact}[theorem]{Fact}
\newtheorem{claim}[theorem]{Claim}
\theoremstyle{definition}
\newtheorem{definition}[theorem]{Definition}
\newtheoremstyle{named}{}{}{\itshape}{}{\bfseries}{.}{.5em}{\thmnote{#3}}
\theoremstyle{named}
\newtheorem{namedtheorem}{Theorem}[section]
\newenvironment{claim_proof}[1][Proof of Claim]
  {\proof[#1]\leftskip=1cm\rightskip=1cm}
  {\endproof}
\makeatletter
\renewcommand\@biblabel[1]{}
\makeatother
\newcommand{\ebar}{\overline{\epsilon}}

\title{The Complexity of Computing the Optimal Composition \\of Differential Privacy}
\author{Jack Murtagh\thanks{Supported by NSF grant CNS-1237235 and a grant from the Sloan Foundation.}\qquad \qquad Salil Vadhan\thanks{Supported by NSF grant CNS-1237235, a grant from the Sloan Foundation, and a Simons Investigator Award.}\\
\\
Center for Research on Computation \& Society\\ John A. Paulson School of Engineering \& Applied Sciences\\ Harvard University\\ \texttt{\{jmurtagh,salil\}@seas.harvard.edu}\\
\url{scholar.harvard.edu/jmurtagh}\\
\url{people.seas.harvard.edu/~salil}}
\date{\today}
\begin{document}
\maketitle
\begin{abstract}
In the study of differential privacy, composition theorems (starting with the original paper of Dwork, McSherry, Nissim, and Smith (TCC'06)) bound the degradation of privacy when composing several differentially private algorithms. Kairouz, Oh, and Viswanath (ICML'15) showed how to compute the optimal bound for composing $k$ arbitrary $(\epsilon,\delta)$-differentially private algorithms. We characterize the optimal composition for the more general case of $k$ arbitrary $(\epsilon_{1},\delta_{1}),\ldots,(\epsilon_{k},\delta_{k})$-differentially private algorithms where the privacy parameters may differ for each algorithm in the composition. We show that computing the optimal composition in general is $\#$P-complete. Since computing optimal composition exactly is infeasible (unless FP=$\#$P), we give an approximation algorithm that computes the composition to arbitrary accuracy in polynomial time. The algorithm is a modification of Dyer's dynamic programming approach to approximately counting solutions to knapsack problems (STOC'03).
\end{abstract}

\*
\begin{section}{Introduction}
Differential privacy is a framework that allows statistical analysis of private databases while minimizing the risks to individuals in the databases. The idea is that an individual should be relatively unaffected whether he or she decides to join or opt out of a research dataset. More specifically, the probability distribution of outputs of a statistical analysis of a database should be nearly identical to the distribution of outputs on the same database with a single person's data removed. Here the probability space is over the coin flips of the randomized differentially private algorithm that handles the queries. To formalize this, we call two databases $D_{0}, D_{1}$ with $n$ rows each $\textit{neighboring}$ if they are identical on at least $n-1$ rows, and define differential privacy as follows:

\begin{definition}[Differential Privacy \cite{DMNS06, DKMMN06}]
A randomized algorithm $M$ is \emph{$(\epsilon,\delta)$-differentially private} for $\epsilon,\delta \geq 0$ if for all pairs of neighboring databases $D_{0}$ and $D_{1}$ and all output sets $S\subseteq \mathrm{Range}(M)$
\[
\Pr[M(D_{0})\in S] \leq e^{\epsilon}\Pr[M(D_{1})\in S] + \delta
\]
where the probabilities are over the coin flips of the algorithm $M$. 
\end{definition}

In the practice of differential privacy, we generally think of $\epsilon$ as a small, non-negligible, constant (e.g. $\epsilon=.1$). We view $\delta$ as a ``security parameter'' that is cryptographically small (e.g. $\delta=2^{-30}$). One of the important properties of differential privacy is that if we run multiple distinct differentially private algorithms on the same database, the resulting composed algorithm is also differentially private, albeit with some degradation in the privacy parameters $(\epsilon,\delta)$. In this paper, we are interested in quantifying the degradation of privacy under composition. We will denote the composition of $k$ differentially private algorithms $M_{1},M_{2},\ldots,M_{k}$ as $(M_{1},M_{2},\ldots,M_{k})$ where
\[
(M_{1},M_{2},\ldots,M_{k})(x)=(M_{1}(x),M_{2}(x),\ldots,M_{k}(x))
\]
A handful of composition theorems already exist in the literature. The first basic result says:

\begin{theorem}[Basic Composition \cite{DKMMN06}]
\label{thm:basiccomp}
For every $\epsilon \geq 0$, $\delta \in [0,1]$, and $(\epsilon,\delta)$-differentially private algorithms $M_{1},M_{2},\ldots,M_{k}$, the composition $(M_{1},M_{2},\ldots,M_{k})$ satisfies $(k\epsilon,k\delta)$-differential privacy.
\end{theorem}

This tells us that under composition, the privacy parameters of the individual algorithms ``sum up,'' so to speak. We care about understanding composition because in practice we rarely want to release only a single statistic about a dataset. Releasing many statistics may require running multiple differentially private algorithms on the same database. Composition is also a very useful tool in algorithm design. Often, new differentially private algorithms are created by combining several simpler algorithms. Composition theorems help us analyze the privacy properties of algorithms designed in this way.

Theorem \ref{thm:basiccomp} shows a linear degradation in global privacy as the number of algorithms in the composition $(k)$ grows and it is of interest to improve on this bound. If we can prove that privacy degrades more slowly under composition, we can get more utility out of our algorithms under the same global privacy guarantees. Dwork, Rothblum, and Vadhan gave the following improvement on the basic summing composition above \cite{DRV10}.

\begin{theorem}[Advanced Composition \cite{DRV10}]
\label{thm:advancedcomp}
For every $\epsilon>0, \delta, \delta'>0,$ $k\in \mathbb{N},$ and $(\epsilon,\delta)$-differentially private algorithms $M_{1},M_{2},\ldots,M_{k}$, the composition $(M_{1},M_{2},\ldots,M_{k})$ satisfies $(\epsilon_{g},k\delta+\delta')$-differential privacy for

\[
\epsilon_{g} = \sqrt{2k\ln(1/\delta')}\cdot\epsilon + k\cdot\epsilon\cdot(e^{\epsilon}-1)
\]
\end{theorem}

Theorem \ref{thm:advancedcomp} shows that privacy under composition degrades by a function of $O(\sqrt{k\ln(1/\delta')})$ which is an improvement if $\delta'=2^{-O(k)}$. It can be shown that a degradation function of $\Omega(\sqrt{k\ln(1/\delta)})$ is necessary even for the simplest differentially private algorithms, such as randomized response \cite{War65}.

Despite giving an asymptotically correct upper bound for the global privacy parameter, $\epsilon_{g}$, Theorem \ref{thm:advancedcomp} is not exact. We want an exact characterization because, beyond being theoretically interesting, constant factors in composition theorems can make a substantial difference in the practice of differential privacy. Furthermore, Theorem \ref{thm:advancedcomp} only applies to ``homogeneous'' composition where each individual algorithm has the same pair of privacy parameters, $(\epsilon,\delta)$ . In practice we often want to analyze the more general case where some individual algorithms in the composition may offer more or less privacy than others. That is, given algorithms $M_{1},M_{2},\ldots,M_{k}$, we want to compute the best achievable privacy parameters for $(M_{1},M_{2},\ldots,M_{k})$. Formally, we want to compute the function:
\[
\mathrm{OptComp}(M_{1},M_{2},\ldots,M_{k},\delta_{g}) = \inf\{\epsilon_{g}\geq 0\colon(M_{1},M_{2},\ldots,M_{k}) ~\textrm{is}~ (\epsilon_{g},\delta_{g})\textsc{-DP}\}
\]

It is convenient for us to view $\delta_{g}$ as given and then compute the best $\epsilon_{g}$, but the dual formulation, viewing $\epsilon_{g}$ as given, is equivalent (by binary search). Actually, we want a function that depends only on the privacy parameters of the individual algorithms:
\[
\mathrm{OptComp}((\epsilon_{1},\delta_{1}),(\epsilon_{2},\delta_{2}),\ldots,(\epsilon_{k},\delta_{k}), \delta_{g}) = \sup\{\mathrm{OptComp}(M_{1},M_{2},\ldots,M_{k},\delta_{g})\colon M_{i} ~\textrm{is}~ (\epsilon_{i},\delta_{i})\textsc{-DP} ~\forall i\in [k]\}
\]

In other words we want $\mathrm{OptComp}$ to give us the minimum possible $\epsilon_{g}$ that maintains privacy for every sequence of algorithms with the given privacy parameters $(\epsilon_{i},\delta_{i})$. A result from Kairouz, Oh, and Viswanath \cite{KOV15} characterizes $\mathrm{OptComp}$ for the homogeneous case.

\begin{theorem}[Optimal Homogeneous Composition \cite{KOV15}\footnote{The phrasing of Theorem \ref{thm:homogeneouscomp} is not exactly how it is presented in \cite{KOV15} (which only refers to $\epsilon_{g}$ of the form $(k-2i)\epsilon$ for integer $i$), but this version can be deduced from the original.}]
\label{thm:homogeneouscomp}
For every $\epsilon \geq 0$ and $\delta \in [0,1)$,\\*
$\mathrm{OptComp}(\underbrace{(\epsilon,\delta),(\epsilon,\delta),\ldots,(\epsilon,\delta)}_{k}, \delta_{g})$ equals the least value of $\epsilon_{g}\geq 0$ such that

\[
 \frac{1}{(1+e^{\epsilon})^{k}}\sum\limits^{k}_{l=\left\lceil{\frac{\epsilon_{g}+k\epsilon}{2\epsilon}}\right\rceil}{k\choose l}\left(e^{l\epsilon}-e^{\epsilon_{g}}e^{(k-l)\epsilon}\right)\leq 1-\frac{1-\delta_{g}}{(1-\delta)^{k}}
\]
\end{theorem}

 Empirically (see Appendix \ref{app:A}), this optimal bound provides a 30-40$\%$ savings in $\epsilon_{g}$ compared to Theorem \ref{thm:advancedcomp} (and a $20\%$ savings compared to an improved asymptotic bound from \cite{KOV15}). The problem remains to find the optimal composition behavior for the more general heterogeneous case. Kairouz, Oh, and Viswanath also provide an upper bound for heterogeneous composition that generalizes the $O(\sqrt{k\ln(1/\delta')})$ degradation found in Theorem \ref{thm:advancedcomp} for homogeneous composition but do not comment on how close it is to optimal. 
\end{section}

\begin{subsection}{Our Results}

We begin by extending the results of Kairouz, Oh, and Viswanath \cite{KOV15} to the general heterogeneous case.

\begin{theorem}[Optimal Heterogeneous Composition]
\label{thm:optcomp}
For all $\epsilon_{1},\ldots,\epsilon_{k} \geq 0$ and $\delta_{1},\ldots,\delta_{k},\delta_{g}\in [0,1), \mathrm{OptComp}((\epsilon_{1},\delta_{1}),(\epsilon_{2},\delta_{2}),\ldots,(\epsilon_{k},\delta_{k}), \delta_{g})$ equals the least value of $\epsilon_{g}\geq 0$ such that

\begin{equation}
\label{eq:1}
\frac{1}{\prod_{i=1}^{k}{(1+e^{\epsilon_{i}})}}
\sum_{S\subseteq \{1,\ldots,k\}}\max\left\{e^{\sum\limits_{i\in S}\epsilon_{i}} - e^{\epsilon_{g}}\cdot e^{\sum\limits_{i\not\in S}\epsilon_{i}}, 0\right\} \leq 1-\frac{1-\delta_{g}}{\prod_{i=1}^{k}{(1-\delta_{i})}} 
\end{equation}
\end{theorem}

Theorem \ref{thm:optcomp} exactly characterizes the optimal composition behavior for any arbitrary set of differentially private algorithms. It also shows that optimal composition can be computed in time exponential in $k$ by computing the sum over $S\subseteq \{1,\ldots,k\}$ by brute force. Of course in practice an exponential-time algorithm is not satisfactory for large $k$. Our next result shows that this exponential complexity is necessary: 

\begin{theorem}
\label{thm:hardness}
Computing $\mathrm{OptComp}$ is $\#P$-complete, even on instances where $\delta_{1}=\delta_{2}=\ldots=\delta_{k}=0$ and $\sum_{i\in [k]}\epsilon_{i}\leq \epsilon$ for any desired constant $\epsilon>0$.
\end{theorem}

Recall that $\#P$ is the class of counting problems associated with decision problems in NP. So being $\#P$-complete means that there is no polynomial-time algorithm for $\mathrm{OptComp}$ unless there is a polynomial-time algorithm for counting the number of satisfying assignments of boolean formulas (or equivalently for counting the number of solutions of all NP problems). So there is almost certainly no efficient algorithm for $\mathrm{OptComp}$ and therefore no analytic solution. Despite the intractability of exact computation, we show that $\mathrm{OptComp}$ can be approximated efficiently.

\begin{theorem}
\label{thm:approx}
There is a polynomial-time algorithm that given rational $\epsilon_{1},\ldots,\epsilon_{k}\geq 0, \delta_{1},\ldots \delta_{k}, \delta_{g} \in [0,1),$ and $\eta\in (0,1)$, outputs $\epsilon^{*}$ satisfying 
\[\mathrm{OptComp}((\epsilon_{1},\delta_{1}),\ldots,(\epsilon_{k},\delta_{k}), \delta_{g})\leq \epsilon^{*}\leq\mathrm{OptComp}((\epsilon_{1},\delta_{1}),\ldots,(\epsilon_{k},\delta_{k}), e^{-\eta/2}\cdot\delta_{g})+\eta\]
The algorithm runs in time \[O\left(\frac{k^{3}\cdot\ebar\cdot(1+\ebar)}{\eta}\cdot \log\left(\frac{k^{2}\cdot\ebar\cdot(1+\ebar)}{\eta}\right)\right)\]
 where $\ebar=\sum_{i\in [k]}\epsilon_{i}/k$, assuming constant-time arithmetic operations.
\end{theorem}

Note that we incur a relative error of $\eta$ in approximating $\delta_{g}$ \emph{and} an additive error of $\eta$ in approximating $\epsilon_{g}$. Since we always take $\epsilon_{g}$ to be non-negligible or even constant, we get a very good approximation when $\eta$ is polynomially small or even a constant. Thus, it is acceptable that the running time is polynomial in $1/\eta$. 

In addition to the results listed above, our proof of Theorem \ref{thm:optcomp} also provides a somewhat simpler proof of the Kairouz-Oh-Viswanath homogeneous composition theorem (Theorem \ref{thm:homogeneouscomp} \cite{KOV15}). The proof in \cite{KOV15} introduces a view of differential privacy through the lens of hypothesis testing and uses geometric arguments. Our proof relies only on elementary techniques commonly found in the differential privacy literature.

\paragraph{Practical Application.} 
The theoretical results presented here were motivated by our work on an applied project called ``Privacy Tools for Sharing Research Data''\footnote{privacytools.seas.harvard.edu}. We are building a system that will allow researchers with sensitive datasets to make differentially private statistics about their data available through data repositories using the Dataverse\footnote{dataverse.org} platform \cite{Cro11, Kin07}. Part of this system is a tool that helps both data depositors and data analysts distribute a global privacy budget across many statistics. Users select which statistics they would like to compute and are given estimates of how accurately each statistic can be computed. They can also redistribute their privacy budget according to which statistics they think are most valuable in their dataset. We implemented the approximation algorithm from Theorem \ref{thm:approx} and integrated it with this tool to ensure that users get the most utility out of their privacy budget. 
\end{subsection}
\begin{section}{Technical Preliminaries}
\label{sect:tech}
A useful notation for thinking about differential privacy is defined below.
\begin{definition}
For two discrete random variables $Y$ and $Z$ taking values in the same output space $S$, the \emph{$\delta$-approximate max-divergence} of $Y$ and $Z$ is defined as:
\[
D_{\infty}^{\delta}(Y\|Z)\equiv \max_{S}\left[\ln\frac{\Pr[Y\in S]-\delta}{\Pr[Z\in S]}\right]
\]
\end{definition}

Notice that an algorithm $M$ is $(\epsilon,\delta)$ differentially private if and only if for all pairs of neighboring databases, $D_{0},D_{1}$, we have $D_{\infty}^{\delta}(M(D_{0})\|M(D_{1}))\leq\epsilon$. The standard fact that differential privacy is closed under ``post processing'' \cite{DMNS06, DR13} now can be formulated as:
\begin{fact}
\label{fact:postproc}
If $f\colon S\to R$ is any randomized function, then 
\[
D_{\infty}^{\delta}(f(Y)\|f(Z))\leq D_{\infty}^{\delta}(Y\|Z)
\]
\end{fact}

\paragraph{Adaptive Composition.} The composition results in our paper actually hold for a more general model of composition than the one described in the introduction. The model is called $k$-fold adaptive composition and was formalized in \cite{DRV10}. We generalize their formulation to the heterogeneous setting where privacy parameters may differ across different algorithms in the composition.

The idea is that instead of running $k$ differentially private algorithms chosen all at once on a single database, we can imagine an adversary adaptively engaging in a ``composition game.'' The game takes as input a bit $b\in\{0,1\}$ and privacy parameters $(\epsilon_{1},\delta_{1}),\ldots,(\epsilon_{k},\delta_{k})$. A randomized adversary $A$, tries to learn $b$ through $k$ rounds of interaction as follows: on the $i$th round of the game, $A$ chooses an $(\epsilon_{i},\delta_{i})$-differentially private algorithm $M_{i}$ and two neighboring databases $D_{(i,0)},D_{(i,1)}$. $A$ then receives an output $y_{i}=M_{i}(D_{(i,b)})$ where the internal randomness of $M_{i}$ is independent of the internal randomness of $M_{1},\ldots,M_{i-1}$. The choices of $M_{i}, D_{(i,0)},$ and $D_{(i,1)}$ may depend on $y_{0},\ldots,y_{i-1}$ as well as the adversary's own randomness.

The outcome of this game is called the \emph{view of the adversary}, $V^{b}$ which is defined to be $(y_{1},\ldots,y_{k})$ along with $A$'s coin tosses. The algorithms $M_{i}$ and databases $D_{(i,0)},D_{(i,1)}$ from each round can be reconstructed from $V^{b}$. Now we can formally define privacy guarantees under $k$-fold adaptive composition.

\begin{definition}
We say that the sequences of privacy parameters $\epsilon_{1},\ldots,\epsilon_{k}\geq 0, \delta_{1},\ldots,\delta_{k}\in [0,1)$ satisfy $(\epsilon_{g},\delta_{g})$-differential privacy \emph{under adaptive composition} if for every adversary $A$ we have $D_{\infty}^{\delta_{g}}(V^{0}\|V^{1})\leq\epsilon_{g}$, where $V^{b}$ represents the view of $A$ in composition game $b$ with privacy parameter inputs $(\epsilon_{1},\delta_{1}),\ldots,(\epsilon_{k},\delta_{k})$.
\end{definition}

\paragraph{Computing real-valued functions.} Many of the computations we discuss involve irrational numbers and we need to be explicit about how we model such computations on finite, discrete machines. Namely when we talk about computing a function $f:\{0,1\}^{*}\to\mathbb{R}$, what we really mean is computing $f$ to any desired number $q$ bits of precision. More precisely, given $x,q$, the task is to compute a number $y\in\mathbb{Q}$ such that $\left|f(x)-y\right|\leq \frac{1}{2^{q}}$. We measure the complexity of algorithms for this task as a function of $|x|+q$. In order to reason about the complexity of OptComp, we will also require that the inputs be rational. So when we talk about computing OptComp exactly, we actually mean given $\epsilon_{1},\ldots,\epsilon_{k}\geq 0, \delta_{1},\ldots,\delta_{k}, \delta_{g}\in [0,1)$ all rational and an integer $q$, compute $\epsilon^{*}$ such that:
\[
\left|\epsilon_{g}-\epsilon^{*}\right|\leq\frac{1}{2^{q}}
\]
where $\epsilon_{g}$ is the true optimal parameter with full precision.
\end{section}

\begin{section}{Characterization of OptComp}

Following \cite{KOV15}, we show that to analyze the composition of arbitrary $(\epsilon_{i},\delta_{i})$-DP algorithms, it suffices to analyze the composition of the following simple variant of randomized response \cite{War65}. 

\begin{definition}[\cite{KOV15}]
Define a randomized algorithm $\tilde{M}_{(\epsilon,\delta)}\colon\{0,1\} \to \{0,1,2,3\}$ as follows, setting $\alpha=1-\delta$:
\begin{center}
\begin{tabular}{ll} 
$\Pr[\tilde{M}_{(\epsilon,\delta)}(0)=0]=\delta$ & $\Pr[\tilde{M}_{(\epsilon,\delta)}(1)=0]=0$  \\
$\Pr[\tilde{M}_{(\epsilon,\delta)}(0)=1]=\alpha\cdot\frac{e^{\epsilon}}{1 + e^{\epsilon}}$ & $\Pr[\tilde{M}_{(\epsilon,\delta)}(1)=1]=\alpha\cdot\frac{1}{1 + e^{\epsilon}}$\\
$\Pr[\tilde{M}_{(\epsilon,\delta)}(0)=2]=\alpha\cdot\frac{1}{1 + e^{\epsilon}}$ & $\Pr[\tilde{M}_{(\epsilon,\delta)}(1)=2]=\alpha\cdot\frac{e^{\epsilon}}{1 + e^{\epsilon}}$\\
$\Pr[\tilde{M}_{(\epsilon,\delta)}(0)=3]=0$ & $\Pr[\tilde{M}_{(\epsilon,\delta)}(1)=3]=\delta$
\end{tabular}
\end{center}
\end{definition}
Note that $\tilde{M}_{(\epsilon,\delta)}$ is in fact $(\epsilon,\delta)$-DP. Kairouz, Oh, and Viswanath showed that $\tilde{M}_{(\epsilon,\delta)}$ can be used to simulate the output of every $(\epsilon,\delta)$-DP algorithm on adjacent databases. 

\begin{lemma}[\cite{KOV15}]
\label{lem:sim}
For every $(\epsilon,\delta)$-DP algorithm $M$ and neighboring databases $D_{0},D_{1}$, there exists a randomized algorithm $T$ such that $T(\tilde{M}_{(\epsilon,\delta)}(b))$ is identically distributed to $M(D_{b})$ for $b=0$ and $b=1$.
\end{lemma}

For the sake of completeness, we provide a self-contained proof of this lemma, which does not use the hypothesis testing and geometric arguments in \cite{KOV15}. Specifically, we give an explicit construction of the simulator, $T$ in two steps. First we introduce a slight generalization of $\tilde{M}_{(\epsilon,\delta)}$ called $\tilde{M}_{(\epsilon,\delta_{0},\delta_{1})}$ and an algorithm $T'$ that can use $\tilde{M}_{(\epsilon,\delta_{0},\delta_{1})}$ to simulate every differentially private algorithm on adjacent databases for some $\delta_{0},\delta_{1}\leq \delta$.  Then we show how to simulate $\tilde{M}_{(\epsilon,\delta_{0},\delta_{1})}$ using $\tilde{M}_{(\epsilon,\delta)}$ with an algorithm called $T''$.  The construction will look like:
\[
\tilde{M}_{(\epsilon,\delta)}(b) \xrightarrow{T''} \tilde{M}_{(\epsilon,\delta_{0},\delta_{1})}(b) \xrightarrow{T'} M(D_{b})
\]
Then the $T$ needed for Lemma \ref{lem:sim} will be $T=T'\circ T''$. Before introducing $\tilde{M}_{(\epsilon,\delta_{0},\delta_{1})}$ and $T'$ we define some additional notation. 

Given an $(\epsilon,\delta)$-DP algorithm $M$ with output space $R$ and neighboring databases $D_{0},D_{1}$, let $P_{0}, P_{1}$ be the probability mass functions of $M(D_{0})$ and $M(D_{1})$, respectively. The definition of differential privacy tells us that for all sets $S\subseteq R$:
\begin{gather*}
P_{0}(S)-e^{\epsilon}P_{1}(S) \leq \delta\\
P_{1}(S)-e^{\epsilon}P_{0}(S) \leq \delta
\end{gather*}
The left-hand side of the first inequality is maximized by $S=S_{0}$ for 
\begin{equation}
\label{eq:2}
S_{0}=\{r\in R\colon P_{0}(r)>e^{\epsilon}P_{1}(r)\}
\end{equation}
and the left-hand side of the second inequality is maximized by 
\begin{equation}
\label{eq:3}
S_{1}=\{r\in R\colon P_{1}(r)>e^{\epsilon}P_{0}(r)\}
\end{equation}
Define $\delta_{0},\delta_{1}$ as
\begin{gather}
\delta_{0}=P_{0}(S_{0})-e^{\epsilon}P_{1}(S_{0}) \leq \delta \label{eq:4}\\
\delta_{1}=P_{1}(S_{1})-e^{\epsilon}P_{0}(S_{1}) \leq \delta \label{eq:5}
\end{gather}
We will show how to simulate $M$ using the following algorithm.
\begin{definition}
Define $\tilde{M}_{(\epsilon,\delta_{0},\delta_{1})}\colon\{0,1\} \to \{0,1,2,3\}$ as follows, with $\delta_{0},\delta_{1}$ as defined in Equations \ref{eq:4} and \ref{eq:5} for some $(\epsilon,\delta)$-DP algorithm and setting $\alpha_{0}=1-\delta_{0}, \alpha_{1}=1-\delta_{1}$:
\begin{center}
\begin{tabular}{ll} 
$\Pr[\tilde{M}_{(\epsilon,\delta_{0},\delta_{1})}(0)=0]=\delta_{0}$ & $\Pr[\tilde{M}_{(\epsilon,\delta_{0},\delta_{1})}(1)=0]=0$  \\
$\Pr[\tilde{M}_{(\epsilon,\delta_{0},\delta_{1})}(0)=1]=\frac{e^{2\epsilon}\alpha_{0}-e^{\epsilon}\alpha_{1}}{e^{2\epsilon}-1}$ & $\Pr[\tilde{M}_{(\epsilon,\delta_{0},\delta_{1})}(1)=1]=\frac{e^{\epsilon}\alpha_{0}-\alpha_{1}}{e^{2\epsilon}-1}$\\
$\Pr[\tilde{M}_{(\epsilon,\delta_{0},\delta_{1})}(0)=2]=\frac{e^{\epsilon}\alpha_{1}-\alpha_{0}}{e^{2\epsilon}-1}$ & $\Pr[\tilde{M}_{(\epsilon,\delta_{0},\delta_{1})}(1)=2]=\frac{e^{2\epsilon}\alpha_{1}-e^{\epsilon}\alpha_{0}}{e^{2\epsilon}-1}$\\
$\Pr[\tilde{M}_{(\epsilon,\delta_{0},\delta_{1})}(0)=3]=0$ & $\Pr[\tilde{M}_{(\epsilon,\delta_{0},\delta_{1})}(1)=3]=\delta_{1}$
\end{tabular}
\end{center}
\end{definition}
Notice that if $\delta_{0}=\delta_{1}=\delta$ then $\tilde{M}_{(\epsilon,\delta_{0},\delta_{1})}=\tilde{M}_{(\epsilon,\delta)}$. We need to show that $\tilde{M}_{(\epsilon,\delta_{0},\delta_{1})}$ is composed of a valid probability distribution. Since $\alpha_{b}=1-\delta_{b}$,
\[\sum_{x\in \{0,1,2,3\}}\Pr[\tilde{M}_{(\epsilon,\delta_{0},\delta_{1})}(b)=x]=1~\mathrm{for}~b=0,1
\]
To see that all of the terms are non-negative we need to show that the recurring terms $e^{\epsilon}\alpha_{1}-\alpha_{0}$ and $e^{\epsilon}\alpha_{0}-\alpha_{1}$ are non-negative and the rest follows by inspection. 
\begin{lemma}
\label{lem:ineq}
For every $(\epsilon,\delta)$-DP algorithm, $M$ with output space $R$ and neighboring databases $D_{0}$ and $D_{1}$, $e^{\epsilon}\alpha_{1}-\alpha_{0}$ and $e^{\epsilon}\alpha_{0}-\alpha_{1}$ are non-negative where $\alpha_{0}=1-\delta_{0}, \alpha_{1}=1-\delta_{1}$ and $\delta_{0},\delta_{1}$ are defined in Equations $\ref{eq:4}$ and $\ref{eq:5}$.
\end{lemma}
\begin{proof}
\begin{align*}
\alpha_{1}&=1-P_{1}(S_{1})+e^{\epsilon}P_{0}(S_{1})\\
&\leq P_{1}(S_{0})+e^{\epsilon}\cdot(1-P_{0}(S_{0}))\\
&\leq e^{2\epsilon}P_{1}(S_{0})+e^{\epsilon}\cdot(1-P_{0}(S_{0}))\\
&=e^{\epsilon}\alpha_{0}
\end{align*}
The other inequality follows by symmetry.
\end{proof}
\noindent Now we show how to use $\tilde{M}_{(\epsilon,\delta_{0},\delta_{1})}$ to simulate any $(\epsilon,\delta)$ differentially private algorithm. 
\begin{lemma}
\label{lem:sim1}
For every $(\epsilon, \delta)$-DP algorithm $M$ with output space $R$, and every pair of neighboring databases, $D_{0}$, $D_{1}$, there exists $\delta_{0},\delta_{1}\leq \delta$ and a randomized algorithm $T'\colon \{0,1,2,3\} \to R$ such that $T'(\tilde{M}_{(\epsilon,\delta_{0},\delta_{1})}(b))$ is identically distributed to $M(D_{b})$ for $b=0$ and $b=1$.
\end{lemma}

\begin{proof}
Fix neighboring databases, $D_{0},D_{1}$ and let $P_{0}, P_{1}$ be the probability mass functions of $M$ on $D_{0},D_{1}$, respectively. We will use $S_{0}, S_{1}, \delta_{0},$ and $\delta_{1}$ as defined above in Equations \ref{eq:2}, \ref{eq:3}, \ref{eq:4}, and \ref{eq:5}. Fix $r\in R$. $T'\colon \{0,1,2,3\}\to R$ is defined in the table below.  
\begin{center}
\[
  \begin{array}{| l | c | c | c |}
    \hline
    x & \Pr[T'(x)=r], r\in S_{0} & \Pr[T'(x)=r], r\in S_{1}& \Pr[T'(x)=r], r\in R\setminus S_{0}\setminus S_{1}\\ \hline
    0 &\frac{1}{\delta_{0}}(P_{0}(r)-e^{\epsilon}P_{1}(r))&0&0 \\ \hline
    1 &\frac{(e^{2\epsilon}-1)P_{1}(r)}{e^{\epsilon}\alpha_{0}-\alpha_{1}}&0&\frac{e^{\epsilon}P_{0}(r)-P_{1}(r)}{e^{\epsilon}\alpha_{0}-\alpha_{1}} \\ \hline
    2 & 0&\frac{(e^{2\epsilon}-1)P_{0}(r)}{e^{\epsilon}\alpha_{1}-\alpha_{0}}&\frac{e^{\epsilon}P_{1}(r)-P_{0}(r)}{e^{\epsilon}\alpha_{1}-\alpha_{0}} \\ \hline
    3 &0&\frac{1}{\delta_{1}}(P_{1}(r)-e^{\epsilon}P_{0}(r))&0 \\ 
    \hline
  \end{array}
\]
\end{center}
We need to show that $T'(x)$ is a valid probability distribution for each $x$. All of the terms are non-negative because $e^{\epsilon}\alpha_{1}-\alpha_{0}$ and $e^{\epsilon}\alpha_{0}-\alpha_{1}$ are non-negative by Lemma \ref{lem:ineq}. 

The sums of $\Pr[T'(0)=r]$ and $\Pr[T'(3)=r]$ are immediate from the definitions of $\delta_{0}$ and $\delta_{1}$, respectively: 
\[
\sum\limits_{r\in R}\Pr[T'(0)=r]=\frac{1}{\delta_{0}}\sum\limits_{r\in S_{0}}(P_{0}(r)-e^{\epsilon}P_{1}(r))+0+0=1 
\]
A symmetrical argument works for $\Pr[T'(3)=r]$. We now analyze the sum for $\Pr[T'(1)=r]$. The sum for $\Pr[T'(2)=r]$ follows by symmetry. We use the following identities:
\begin{gather*}
\alpha_{0}=1-\sum\limits_{r\in S_{0}}(P_{0}(r)-e^{\epsilon}P_{1}(r))=\sum\limits_{r\in S_{0}}e^{\epsilon}P_{1}(r) +\sum\limits_{r\in S_{1}}P_{0}(r) + \sum\limits_{r\in R\setminus S_{0}\setminus S_{1}}P_{0}(r)\\
\alpha_{1}=1-\sum\limits_{r\in S_{1}}(P_{1}(r)-e^{\epsilon}P_{0}(r))=\sum\limits_{r\in S_{0}}P_{1}(r)+\sum\limits_{r\in S_{1}}e^{\epsilon}P_{0}(r) +\sum\limits_{r\in R\setminus S_{0}\setminus S_{1}}P_{1}(r)
\end{gather*}
Thus:
\begin{equation*}
e^{\epsilon}\alpha_{0}-\alpha_{1}=\sum\limits_{r\in S_{0}}(e^{2\epsilon}-1)P_{1}(r)+\sum\limits_{r\in R\setminus S_{0}\setminus S_{1}}(e^{\epsilon}P_{0}(r)-P_{1}(r))
\end{equation*}
This implies $\sum\limits_{r\in R}\Pr[T'(1)=r]=1$. Now we just need to show that $T'(\tilde{M}_{(\epsilon,\delta_{0},\delta_{1})}(b))$ is identically distributed to $M(D_{b})$. We will show this for $b=0$ and the $b=1$ case follows by symmetry. Fix $r\in R$. By the definition of $\tilde{M}_{(\epsilon,\delta_{0},\delta_{1})}$:
\[
\Pr[T'(\tilde{M}_{(\epsilon,\delta_{0},\delta_{1})}(0))=r]=\delta_{0}\cdot \Pr[T'(0)=r]+\left(\frac{e^{2\epsilon}\alpha_{0}-e^{\epsilon}\alpha_{1}}{e^{2\epsilon}-1}\right)\cdot \Pr[T'(1)=r]+\left(\frac{e^{\epsilon}\alpha_{1}-\alpha_{0}}{e^{2\epsilon}-1}\right)\cdot \Pr[T'(2)=r] 
\]
From here we break the calculation into the three possible cases:

\vspace*{\baselineskip}
$\textbf{Case 1:}$ $r\in S_{0}$
\begin{align*}
\Pr[T'(\tilde{M}_{(\epsilon,\delta_{0},\delta_{1})}(0))=r]&=\delta_{0}\cdot\left(\frac{1}{\delta_{0}}(P_{0}(r)-e^{\epsilon}P_{1}(r))\right)+\frac{e^{2\epsilon}\alpha_{0}-e^{\epsilon}\alpha_{1}}{e^{2\epsilon}-1}\cdot\frac{(e^{2\epsilon}-1)P_{1}(r)}{e^{\epsilon}\alpha_{0}-\alpha_{1}}\\
&=P_{0}(r)-e^{\epsilon}P_{1}(r)+e^{\epsilon}P_{1}(r)=P_{0}(r)
\end{align*}

\vspace*{\baselineskip}
$\textbf{Case 2:}$ $r\in S_{1}$
\begin{align*}
\Pr[T'(\tilde{M}_{(\epsilon,\delta_{0},\delta_{1})}(0))=r]=\frac{e^{\epsilon}\alpha_{1}-\alpha_{0}}{e^{2\epsilon}-1}\cdot\frac{(e^{2\epsilon}-1)P_{0}(r)}{e^{\epsilon}\alpha_{1}-\alpha_{0}}=P_{0}(r)
\end{align*}

\vspace*{\baselineskip}
$\textbf{Case 3:}$ $r\in R\setminus S_{0}\setminus S_{1}$
\begin{align*}
\Pr[T'(\tilde{M}_{(\epsilon,\delta_{0},\delta_{1})}(0))=r]&=\frac{e^{2\epsilon}\alpha_{0}-e^{\epsilon}\alpha_{1}}{e^{2\epsilon}-1}\cdot\frac{e^{\epsilon}P_{0}(r)-P_{1}(r)}{e^{\epsilon}\alpha_{0}-\alpha_{1}} +\frac{e^{\epsilon}\alpha_{1}-\alpha_{0}}{e^{2\epsilon}-1}\cdot\frac{e^{\epsilon}P_{1}(r)-P_{0}(r)}{e^{\epsilon}\alpha_{1}-\alpha_{0}}\\
&=\frac{e^{2\epsilon}P_{0}(r)-e^{\epsilon}P_{1}(r)+e^{\epsilon}P_{1}(r)-P_{0}(r)}{e^{2\epsilon}-1}=P_{0}(r)
\end{align*}
\end{proof}

We have shown how a generalization of $\tilde{M}_{(\epsilon,\delta)}$ called $\tilde{M}_{(\epsilon,\delta_{0},\delta_{1})}$ can be used to simulate the output of every differentially private algorithm. In the next lemma we show how to simulate $\tilde{M}_{(\epsilon,\delta_{0},\delta_{1})}$ using $\tilde{M}_{(\epsilon,\delta)}$, which implies that $\tilde{M}_{(\epsilon,\delta)}$ can be used to simulate the output of every differentially private algorithm by composing the simulator introduced in Lemma \ref{lem:sim1} with the one introduced below. 
\begin{lemma}
\label{lem:sim2}
For every $\epsilon\geq 0$ and $\delta_{0},\delta_{1},\delta\in [0,1)$ such that $e^{\epsilon}\cdot(1-\delta_{0})\geq 1-\delta_{1}$ and $e^{\epsilon}\cdot(1-\delta_{1})\geq 1-\delta_{0} $ and $\delta_{0},\delta_{1}\leq \delta$, there exists a randomized algorithm $T''$ such that $T''(\tilde{M}_{(\epsilon,\delta)}(b))$ is identically distributed to $\tilde{M}_{(\epsilon,\delta_{0},\delta_{1})}(b)$ for both $b=0,1$.
\end{lemma}

\begin{proof}
Assume without loss of generality that $\delta_{0}\geq\delta_{1}$ and set $\alpha=1-\delta, \alpha_{0}=1-\delta_{0},$ and $\alpha_{1}=1-\delta_{1}$. We will represent $T''(\tilde{M}_{(\epsilon,\delta)}(b))$ as a Markov Chain below. Here, the probability of transitioning from one state to another is proportional to the weight of an edge. That is, the true probability along an edge leaving some node $a$ is the weight divided by the sum of the weights of all of the edges leaving $a$ (this is just to avoid cluttering the diagram with the normalizing denominators). 
\begin{center}
\begin{tikzpicture}[->, >=stealth', auto, semithick, node distance=3cm]
\tikzstyle{every state}=[fill=white,draw=black,thick,text=black,scale=1]
\node[state]    (0)                     {$0$};
\node[state]    (1) at ($ (0) + (270:3) $)  {$1$};
\node[draw=black, circle, minimum size=.5cm]    (emp1) at ($ (0) + (0:3) $)  {};
\node[draw=black, circle, minimum size=.5cm]     (emp2) at ($ (1) + (0:3) $)   {};
\node[state]    (M1) at ($ (emp1) + (0:3) $)  {$1$};
\node[state]    (M2) at ($ (emp2) + (0:3) $)   {$2$};
\node[state]    (M0) at ($ (M1) + (90:3) $)  {$0$};
\node[state]    (M3) at ($ (M2) + (270:3) $)   {$3$};
\node[state]    (T0) at ($ (M0) + (0:6) $)   {$0$};
\node[state]    (T1) at ($ (M1) + (0:6) $)   {$1$};
\node[state]    (T2) at ($ (M2) + (0:6) $)   {$2$};
\node[state]    (T3) at ($ (M3) + (0:6) $)   {$3$};
\node[draw=black, circle, minimum size=.5cm]    (emp3) at ($ (M1) + (26.565:3.354) $)   {};
\node[draw=black, circle, minimum size=.5cm]    (emp4) at ($ (M3) + (26.565:3.354) $)   {};
\node[fill=white,draw=white,text=black,scale=1.5,circle, minimum size=1cm] (blab) at ($ (0) + (90:5) $) {$b$};
\node[fill=white,draw=white,text=black,scale=1.5,circle, minimum size=1cm] (Mlab) at ($ (M0) + (90:2) $) {$\tilde{M}_{(\epsilon,\delta)}(b)$};
\node[fill=white,draw=white,text=black,scale=1.5,rectangle, minimum size=1cm] (Tlab) at ($ (T0) + (90:2) $) {$T''(\tilde{M}_{(\epsilon,\delta)}(b))$};

\path[every node/.style={sloped,anchor=south,auto=false}]
(0)    edge[]  node{$\alpha$}         (emp1)
(0)    edge[]  node{$\delta$}         (M0)
(1)    edge[]  node{$\alpha$}         (emp2)
(1)    edge[]  node{$\delta$}         (M3)
(emp1) edge[]  node{$e^{\epsilon}$}   (M1)
(emp1) edge[]  node[pos=.4]{$1$}              (M2)
(emp2) edge[]  node[pos=.4]{$1$}              (M1)
(emp2) edge[]  node{$e^{\epsilon}$}   (M2)
(M0)   edge[]  node{$\delta_{0}$}         (T0)
(M0)   edge[]  node[pos=.57]{$\delta-\delta_{0}=\alpha_{0}-\alpha$}         (emp3)
(M1)   edge[bend left=10]  node[pos=.4]{$e^{\epsilon}\alpha_{0}-\alpha_{1}$}   (T1)
(M1)   edge[]  node{$\alpha_{1}-\alpha_{0}$}              (T2)
(M2)   edge[]  node{$1$}              (T2)
(M3)   edge[]  node{$\delta-\delta_{1}=\alpha_{1}-\alpha$}              (emp4)
(M3)   edge[]  node{$\delta_{1}$}   (T3)
(emp3) edge[]  node{$e^{\epsilon}\alpha_{0}-\alpha_{1}$}              (T1)
(emp3) edge[]  node{$\alpha_{1}-e^{-\epsilon}\alpha_{0}$}              (T2)
(emp4) edge[]  node[pos=.43]{$p$}              (T1)
(emp4) edge[]  node{$1-p$}   (T2);
\end{tikzpicture}
\end{center}

Where 
\[
p=\left(\frac{\alpha_{0}-\alpha}{\alpha_{1}-\alpha}\right)\cdot\left(\frac{e^{\epsilon}\alpha_{0}-\alpha_{1}}{\alpha_{0}(e^{2\epsilon}-1)}\right)
\]

All of the weights are non-negative because $\alpha_{1}\geq\alpha_{0}\geq\alpha$, $e^{\epsilon}\alpha_{1}\geq\alpha_{0}$, and $p$ is also at most $1$, which we verify now:
\begin{align*}
(\alpha_{0}-\alpha)\cdot(e^{\epsilon}\alpha_{0}-\alpha_{1})&\leq(\alpha_{1}-\alpha)\cdot(e^{2\epsilon}\alpha_{0}-\alpha_{1})\\
&\leq(\alpha_{1}-\alpha)\cdot(e^{2\epsilon}\alpha_{0}-\alpha_{0})\\
&=(\alpha_{1}-\alpha)\cdot\alpha_{0}\cdot(e^{2\epsilon}-1)
\end{align*}

We need to show that $T''(\tilde{M}_{(\epsilon,\delta)}(b))$ is identically distributed to $\tilde{M}_{(\epsilon,\delta_{0},\delta_{1})}(b)$ for $b=0$ and $b=1$, which will complete the proof. Notice that $\Pr[T''(\tilde{M}_{(\epsilon,\delta)}(0))=3]=0=\Pr[\tilde{M}_{(\epsilon,\delta_{0},\delta_{1})}(0)=3]$ because there is no path from the $b=0$ node to the $T''=3$ node. Similarly, $\Pr[T''(\tilde{M}_{(\epsilon,\delta)}(1))=0]=0=\Pr[\tilde{M}_{(\epsilon,\delta_{0},\delta_{1})}(1)=0]$  We also have:
\begin{align*}
\Pr[T''(\tilde{M}_{(\epsilon,\delta)}(0))=0]&=\left(\frac{\delta}{\delta+\alpha}\right)\cdot\left(\frac{\delta_{0}}{\delta_{0}+(\delta-\delta_{0})}\right)\\
&=\frac{\delta}{1}\cdot\frac{\delta_{0}}{\delta}
\\&=\delta_{0}\\
&=\Pr[\tilde{M}_{(\epsilon,\delta_{0},\delta_{1})}(0)=0]
\end{align*}
Similarly,
\begin{align*}
\Pr[T''(\tilde{M}_{(\epsilon,\delta)}(3))=3]&=\left(\frac{\delta}{\delta+\alpha}\right)\cdot\left(\frac{\delta_{1}}{\delta_{1}+(\delta-\delta_{1})}\right)\\
&=\frac{\delta}{1}\cdot\frac{\delta_{1}}{\delta}
\\
&=\delta_{1}\\
&=\Pr[\tilde{M}_{(\epsilon,\delta_{0},\delta_{1})}(3)=3]
\end{align*}
Next we check the probabilities with which $T''$ outputs $1$ and $2$ when $b=0$.

\begin{align*}
\Pr[T''(\tilde{M}_{(\epsilon,\delta)}(0))=1]&=\delta\cdot\left(\frac{\alpha_{0}-\alpha}{\delta}\right)\cdot\left(\frac{e^{\epsilon}\alpha_{0}-\alpha_{1}}{\alpha_{0}(e^{\epsilon}-e^{-\epsilon})}\right)+
\alpha\cdot\left(\frac{e^{\epsilon}}{e^{\epsilon}+1}\right)\cdot\left(\frac{e^{\epsilon}\alpha_{0}-\alpha_{1}}{\alpha_{0}(e^{\epsilon}-1)}\right)\\
&=(\alpha_{0}-\alpha+\alpha)\cdot\left(\frac{e^{2\epsilon}\alpha_{0}-e^{\epsilon}\alpha_{1}}{\alpha_{0}(e^{2\epsilon}-1)}\right) \\
&=\frac{e^{2\epsilon}\alpha_{0}-e^{\epsilon}\alpha_{1}}{e^{2\epsilon}-1}\\
&=\Pr[\tilde{M}_{(\epsilon,\delta_{0},\delta_{1})}(0)=1]\\
\end{align*}
It follows that $\Pr[T''(\tilde{M}_{(\epsilon,\delta)}(0))=2]=\Pr[\tilde{M}_{(\epsilon,\delta_{0},\delta_{1})}(0)=2]$ because the probabilities sum to $1$. Finally we show the probabilities with which $T''$ outputs $1$ and $2$ when $b=1$.
\begin{align*}
\Pr[T''(\tilde{M}_{(\epsilon,\delta)}(1))=1]&=\delta\cdot\frac{\alpha_{1}-\alpha}{\delta}\cdot \left(\frac{\alpha_{0}-\alpha}{\alpha_{1}-\alpha}\right)\cdot\left(\frac{e^{\epsilon}\alpha_{0}-\alpha_{1}}{\alpha_{0}(e^{2\epsilon}-1)}\right) + \alpha\cdot\left(\frac{1}{e^{\epsilon}+1}\right)\cdot\left(\frac{e^{\epsilon}\alpha_{0}-\alpha_{1}}{\alpha_{0}(e^{\epsilon}-1)}\right)\\
&=(\alpha_{0}-\alpha+\alpha)\cdot\left(\frac{e^{\epsilon}\alpha_{0}-\alpha_{1}}{\alpha_{0}(e^{2\epsilon}-1)}\right)\\
&=\Pr[\tilde{M}_{(\epsilon,\delta_{0},\delta_{1})}(1)=1]
\end{align*}
Again because the probabilities sum to 1, it follows that $\Pr[T''(\tilde{M}_{(\epsilon,\delta)}(1))=2]=\Pr[\tilde{M}_{(\epsilon,\delta_{0},\delta_{1})}(1)=2]$, which completes the proof.
\end{proof}

So $\tilde{M}_{(\epsilon,\delta)}$ can simulate any $(\epsilon,\delta)$ differentially private algorithm. Since it is known that post-processing preserves differential privacy (Fact \ref{fact:postproc}), it follows that to analyze the composition of arbitrary differentially private algorithms, it suffices to analyze the composition of $\tilde{M}_{(\epsilon_{i},\delta_{i})}$'s:

\begin{lemma}
\label{lem:optism}
For all $\epsilon_{1},\ldots,\epsilon_{k}\geq 0, \delta_{1},\ldots,\delta_{k},\delta_{g}\in [0,1)$,
\[
\mathrm{OptComp}((\epsilon_{1},\delta_{1}),\ldots,(\epsilon_{k},\delta_{k}), \delta_{g})= \mathrm{OptComp}(\tilde{M}_{(\epsilon_{1},\delta_{1})},\ldots,\tilde{M}_{(\epsilon_{k},\delta_{k})},\delta_{g})
\]
\end{lemma}
\begin{proof}
Since $\tilde{M}_{(\epsilon_{1},\delta_{1})},\ldots,\tilde{M}_{(\epsilon_{k},\delta_{k})}$ are $(\epsilon_{1},\delta_{1}),\ldots,(\epsilon_{k},\delta_{k})$-differentially private, we have:
\begin{align*}
\mathrm{OptComp}((\epsilon_{1},\delta_{1}),\ldots,(\epsilon_{k},\delta_{k}), \delta_{g})&=\sup\{\mathrm{OptComp}(M_{1},\ldots,M_{k},\delta_{g})\colon M_{i} ~\textrm{is}~ (\epsilon_{i},\delta_{i})\textsc{-DP} ~\forall i\in [k]\}\\
&\geq \mathrm{OptComp}(\tilde{M}_{(\epsilon_{1},\delta_{1})},\ldots,\tilde{M}_{(\epsilon_{k},\delta_{k})},\delta_{g})
\end{align*}

For the other direction, it suffices to show that for every $M_{1},\ldots,M_{k}$ that are $(\epsilon_{1},\delta_{1}),\ldots,(\epsilon_{k},\delta_{k})$-differentially private, we have 
\[\mathrm{OptComp}(M_{1},\ldots,M_{k},\delta_{g})\leq\mathrm{OptComp}(\tilde{M}_{(\epsilon_{1},\delta_{1})},\ldots,\tilde{M}_{(\epsilon_{k},\delta_{k})})\]
That is, 
\[
\inf\{\epsilon_{g}\geq0\colon(M_{1},\ldots,M_{k}) ~\textrm{is}~ (\epsilon_{g},\delta_{g})\textsc{-DP}\}\leq\inf\{\epsilon_{g}\geq 0\colon(\tilde{M}_{(\epsilon_{1},\delta_{1})},\ldots,\tilde{M}_{(\epsilon_{k},\delta_{k})}) ~\textrm{is}~ (\epsilon_{g},\delta_{g})\textsc{-DP}\}
\]
So suppose $(\tilde{M}_{(\epsilon_{1},\delta_{1})},\ldots,\tilde{M}_{(\epsilon_{k},\delta_{k})})$ is $(\epsilon_{g},\delta_{g})$-DP. We will show that $(M_{1},\ldots,M_{k})$ is also $(\epsilon_{g},\delta_{g})$-DP. Taking the infimum over $\epsilon_{g}$ then completes the proof.

We know from Lemma \ref{lem:sim} that for every pair of neighboring databases $D_{0},D_{1}$, there must exist randomized algorithms $T_{1},\ldots,T_{k}$ such that $T_{i}(\tilde{M}_{(\epsilon_{i},\delta_{i})}(b))$ is identically distributed to $M_{i}(D_{b})$ for all $i\in \{1,\ldots,k\}$. By hypothesis we have 
\[D_{\infty}^{\delta_{g}}\left((\tilde{M}_{(\epsilon_{1},\delta_{1})}(0),\ldots,\tilde{M}_{(\epsilon_{k},\delta_{k})}(0))\|(\tilde{M}_{(\epsilon_{1},\delta_{1})}(1),\ldots,\tilde{M}_{(\epsilon_{k},\delta_{k})}(1))\right)\leq\epsilon_{g}
\] 
Thus by Fact \ref{fact:postproc} we have:
\begin{align*}
&D_{\infty}^{\delta_{g}}\big((M_{1}(D_{0}),\ldots,M_{k}(D_{0}))\|(M_{1}(D_{1}),\ldots,M_{k}(D_{1}))\big)=\\
&D_{\infty}^{\delta_{g}}\left((T_{1}(\tilde{M}_{(\epsilon_{1},\delta_{1})}(0)),\ldots,T_{k}(\tilde{M}_{(\epsilon_{k},\delta_{k})}(0)))\|(T_{1}(\tilde{M}_{(\epsilon_{1},\delta_{1})}(1)),\ldots,T_{k}(\tilde{M}_{(\epsilon_{k},\delta_{k})}(1)))\right)\leq\epsilon_{g}
\end{align*}
\end{proof}

Now we are ready to characterize $\mathrm{OptComp}$ for an arbitrary set of differentially private algorithms. 
\begin{proof}[Proof of Theorem $\ref{thm:optcomp}$] 
Given $(\epsilon_{1},\delta_{1}),\ldots,(\epsilon_{k},\delta_{k})$ and $\delta_{g}$, let $\tilde{M}^{k}(b)$ denote the composition\\* $(\tilde{M}_{(\epsilon_{1},\delta_{1})}(b),\ldots,\tilde{M}_{(\epsilon_{k},\delta_{k})}(b))$ and let $\tilde{P}_{b}^{k}(x)$ be the probability mass function of $\tilde{M}^{k}(b)$, for $b=0$ and $b=1$. By Lemma \ref{lem:optism}, $\mathrm{OptComp((\epsilon_{1},\delta_{1}),\ldots,(\epsilon_{k},\delta_{k}),\delta_{g})}$ is the smallest value of $\epsilon_{g}$ such that:

\[
\delta_{g} \geq \max_{Q\subseteq\{0,1,2,3\}^{k}}\left\{\tilde{P}_{0}^{k}(Q)-e^{\epsilon_{g}}\cdot\tilde{P}_{1}^{k}(Q), \tilde{P}_{1}^{k}(Q)-e^{\epsilon_{g}}\cdot\tilde{P}_{0}^{k}(Q)\right\}.
\]
Since $\tilde{M}$ is symmetric, we can instead consider the smallest value of $\epsilon_{g}$ such that:

\[
\delta_{g} \geq \max_{Q\subseteq\{0,1,2,3\}^{k}}\left\{\tilde{P}_{0}^{k}(Q)-e^{\epsilon_{g}}\cdot\tilde{P}_{1}^{k}(Q)\right\},
\]
without loss of generality. Given $\epsilon_{g}$, the set $S\subseteq \{0,1,2,3\}^{k}$ that maximizes the right-hand side is

\[
S=S(\epsilon_{g})=\left\{x\in\{0,1,2,3\}^{k}\mid \tilde{P}_{0}^{k}(x)\geq e^{\epsilon_{g}}\cdot\tilde{P}_{1}^{k}(x)\right\}
\]

We can further split $S(\epsilon_{g})$ into $S(\epsilon_{g}) = S_{0}(\epsilon_{g})\cup S_{1}(\epsilon_{g})$ with

\begin{align*}
S_{0}(\epsilon_{g}) &= \left\{x\in \{0,1,2,3\}^{k} \mid \tilde{P}_{1}^{k}(x) = 0\right\} \\
S_{1}(\epsilon_{g}) &= \left\{x\in \{0,1,2,3\}^{k} \mid \tilde{P}_{0}^{k}(x) \geq e^{\epsilon_{g}}\cdot\tilde{P}_{1}^{k}(x), ~\textrm{and}~ \tilde{P}_{1}^{k}(x) > 0\right\}
\end{align*}

Note that $S_{0}(\epsilon_{g})\cap S_{1}(\epsilon_{g})=\emptyset$. We have $\tilde{P}_{1}^{k}(S_{0}(\epsilon_{g}))=0$ and $\tilde{P}_{0}^{k}(S_{0}(\epsilon_{g})) = 1-\Pr[\tilde{M}^{k}(0)\in \{1,2,3\}^{k}] = 1-\prod_{i=1}^{k}(1-\delta_{i})$. So 

\begin{align*}
\tilde{P}_{0}^{k}(S(\epsilon_{g})) - e^{\epsilon_{g}}\tilde{P}_{1}^{k}(S(\epsilon_{g})) &=  \tilde{P}_{0}^{k}(S_{0}(\epsilon_{g})) - e^{\epsilon_{g}}\tilde{P}_{1}^{k}(S_{0}(\epsilon_{g})) + \tilde{P}_{0}^{k}(S_{1}(\epsilon_{g})) - e^{\epsilon_{g}}\tilde{P}_{1}^{k}(S_{1}(\epsilon_{g})) \\
&= 1-\prod_{i=1}^{k}(1-\delta_{i}) + \tilde{P}_{0}^{k}(S_{1}(\epsilon_{g})) - e^{\epsilon_{g}}\tilde{P}_{1}^{k}(S_{1}(\epsilon_{g}))
\end{align*}

Now we just need to analyze $\tilde{P}_{0}^{k}(S_{1}(\epsilon_{g})) - e^{\epsilon_{g}}\tilde{P}_{1}^{k}(S_{1}(\epsilon_{g}))$. Notice that $S_{1}(\epsilon_{g})\subseteq \{1,2\}^{k}$ because for all $x\in S_{1}(\epsilon_{g})$, we have $\tilde{P}_{0}(x)>\tilde{P}_{1}(x)>0$. So we can write:
\begin{align*}
\tilde{P}_{0}^{k}(S_{1}(\epsilon_{g}))&-e^{\epsilon_{g}}\cdot\tilde{P}_{1}^{k}(S_{1}(\epsilon_{g}))\\
&= \sum_{y\in \{1,2\}^{k}}\max\left\{\prod_{i\colon y_{i}=1}\frac{(1-\delta_{i})e^{\epsilon_{i}}}{1+e^{\epsilon_{i}}}\cdot\prod_{i\colon y_{i}=2}\frac{(1-\delta_{i})}{1+e^{\epsilon_{i}}} - e^{\epsilon_{g}}\prod_{i\colon y_{i}=1}\frac{(1-\delta_{i})}{1+e^{\epsilon_{i}}}\cdot\prod_{i\colon y_{i}=2}\frac{(1-\delta_{i})e^{\epsilon_{i}}}{1+e^{\epsilon_{i}}},0\right\}\\
&=\prod_{i=1}^{k}\frac{1-\delta_{i}}{1+e^{\epsilon_{i}}}\sum_{y\in\{0,1\}^{k}}\max\left\{\frac{e^{\sum_{i=1}^{k}{{\epsilon_{i}}}}}
{e^{\sum_{i=1}^{k}y_{i}\epsilon_{i}}} - e^{\epsilon_{g}}\cdot e^{\sum_{i=1}^{k}y_{i}\epsilon_{i}}, 0\right\}
\end{align*}
Putting everything together yields:
\begin{align*}
\delta_{g} &\geq \tilde{P}_{0}^{k}(S_{0}(\epsilon_{g})) - e^{\epsilon_{g}}\tilde{P}_{1}^{k}(S_{0}(\epsilon_{g})) + \tilde{P}_{0}^{k}(S_{1}(\epsilon_{g})) - e^{\epsilon_{g}}\tilde{P}_{1}^{k}(S_{1}(\epsilon_{g})) \\
&= 1-\prod_{i=1}^{k}(1-\delta_{i}) + \frac{\prod_{i=1}^{k}(1-\delta_{i})}{\prod_{i=1}^{k}{(1+e^{\epsilon_{i}})}}
\sum_{S\subseteq \{1,\ldots,k\}}\max\left\{e^{\sum\limits_{i\in S}\epsilon_{i}} - e^{\epsilon_{g}}\cdot e^{\sum\limits_{i\not\in S}\epsilon_{i}}, 0\right\} 
\end{align*}
\end{proof}

We have characterized the optimal composition for an arbitrary set of differentially private algorithms $(M_{1},\ldots,M_{k})$ under the assumption that the algorithms are chosen in advance and all run on the same database. Next we show that $\mathrm{OptComp}$ under this restrictive model of composition is actually equivalent under the more general $k$-fold adaptive composition discussed in Section \ref{sect:tech}.

\begin{theorem}
\label{thm:optadaptive}
The privacy parameters $\epsilon_{1},\ldots,\epsilon_{k}\geq 0, \delta_{1},\ldots,\delta_{k}\in [0,1)$, satisfy $(\epsilon_{g},\delta_{g})$-differential privacy under adaptive composition for $\epsilon_{g},\delta_{g}\geq 0$ if and only if $\mathrm{OptComp}((\epsilon_{1},\delta_{1}),\ldots,(\epsilon_{k},\delta_{k}),\delta_{g})\leq\epsilon_{g}$
\end{theorem}
\begin{proof}
First suppose the privacy parameters $\epsilon_{1},\ldots,\epsilon_{k},\delta_{1},\ldots,\delta_{k}$ satisfy $(\epsilon_{g},\delta_{g})$-differential privacy under adaptive composition. Then $\mathrm{OptComp}((\epsilon_{1},\delta_{1}),\ldots,(\epsilon_{k},\delta_{k}),\delta_{g})\leq \epsilon_{g}$ because adaptive composition is more general than the composition defining $\mathrm{OptComp}$.

Conversely, suppose $\mathrm{OptComp}((\epsilon_{1},\delta_{1}),\ldots,(\epsilon_{k},\delta_{k}),\delta_{g})\leq\epsilon_{g}$. In particular, this means \\* $\mathrm{OptComp}(\tilde{M}_{(\epsilon_{1},\delta_{1})},\ldots,\tilde{M}_{(\epsilon_{k},\delta_{k})},\delta_{g})\leq\epsilon_{g}$. To complete the proof, we must show that the privacy parameters $\epsilon_{1},\ldots,\epsilon_{k},\delta_{1},\ldots,\delta_{k}$ satisfy $(\epsilon_{g},\delta_{g})$-differential privacy under adaptive composition.

Fix an adversary $A$. On each round $i$, $A$ uses its coin tosses $r$ and the previous outputs $y_{1},\ldots,y_{i-1}$ to select an $(\epsilon_{i},\delta_{i})$-differentially private algorithm $M_{i}=M_{i}^{r,y_{1},\ldots,y_{i-1}}$ and neighboring databases $D_{0}=D_{0}^{r,y_{1},\ldots,y_{i-1}},D_{1}=D_{1}^{r,y_{1},\ldots,y_{i-1}}$. Let $V^{b}$ be the view of $A$ with the given privacy parameters under composition game $b$ for $b=0$ and $b=1$.

Lemma \ref{lem:sim} tells us that there exists an algorithm $T_{i}=T_{i}^{r,y_{1},\ldots,y_{i-1}}$ such that $T_{i}(\tilde{M}_{(\epsilon_{i},\delta_{i})}(b))$ is identically distributed to $M_{i}(D_{b})$ for both $b=0,1$ for all $i\in [k]$. Define $\hat{T}(z_{1},\ldots,z_{k})$ for $z_{1},\ldots,z_{k}\in\{0,1,2,3\}$ as follows:

\begin{enumerate}
\item Randomly choose coins $r$ for $A$
\item For $i=1,\ldots,k,$ let $y_{i}\leftarrow T_{i}^{r,y_{1},\ldots,y_{i-1}}(z_{i})$
\item Output $(r,y_{1},\ldots,y_{k})$ 
\end{enumerate}

Notice that $\hat{T}(\tilde{M}_{(\epsilon_{1},\delta_{1})}(b),\ldots,\tilde{M}_{(\epsilon_{k},\delta_{k})}(b))$ is identically distributed to $V^{b}$ for both $b=0,1$. By hypothesis we have 
\[D_{\infty}^{\delta_{g}}\left((\tilde{M}_{(\epsilon_{1},\delta_{1})}(0),\ldots,\tilde{M}_{(\epsilon_{k},\delta_{k})}(0))\|(\tilde{M}_{(\epsilon_{1},\delta_{1})}(1),\ldots,\tilde{M}_{(\epsilon_{k},\delta_{k})}(1))\right)\leq\epsilon_{g}
\] 
Thus by Fact \ref{fact:postproc} we have:
\[
D_{\infty}^{\delta_{g}}\big(V^{0}\|V^{1}\big)=D_{\infty}^{\delta_{g}}\left(\hat{T}(\tilde{M}_{(\epsilon_{1},\delta_{1})}(0),\ldots,\tilde{M}_{(\epsilon_{k},\delta_{k})}(0))\|\hat{T}(\tilde{M}_{(\epsilon_{1},\delta_{1})}(1),\ldots,\tilde{M}_{(\epsilon_{k},\delta_{k})}(1))\right)\leq\epsilon_{g}
\]
\end{proof}
\end{section}

\begin{section}{Hardness of OptComp}

$\#P$ is the class of all counting problems associated with decision problems in NP. It is a set of functions that count the number of solutions to some NP problem. More formally:

\begin{definition}
A function $f\colon \{0,1\}^{*}\to\mathbb{N}$ is in the class $\#P$ if there exists a polynomial $p\colon \mathbb{N}\to\mathbb{N}$ and a polynomial time algorithm $M$ such that for every $x\in\{0,1\}^{*}$:
\[
f(x)=\left|\left\{y\in\{0,1\}^{p(|x|)}\colon M(x,y)=1\right\}\right|
\]
\end{definition}

\begin{definition}
A function $g$ is called $\#P$\emph{-hard} if every function $f\in\#P$ can be computed in polynomial time given oracle access to $g$. That is, evaluations of $g$ can be done in one time step.
\end{definition}

If a function is $\#P$-hard, then there is no polynomial-time algorithm for computing it unless there is a polynomial-time algorithm for counting the number of solutions of all NP problems.

\begin{definition}
A function $f$ is called $\#P$\emph{-easy} if there is some function $g\in\#P$ such that $f$ can be computed in polynomial time given oracle access to $g$.
\end{definition}

If a function is both $\#P$-hard and $\#P$-easy, we say it is $\#P$-complete. Proving that computing $\textrm{OptComp}$ is $\#P$-complete can be broken into two steps: showing that it is $\#P$-easy and showing that it is $\#P$-hard. 

\begin{lemma}
Computing $\textrm{OptComp}$ is $\#P$-easy.
\end{lemma}
\begin{proof}
For convenience we will view rational $(\epsilon_{1},\delta_{1}),\ldots,(\epsilon_{k},\delta_{k})$ and $\epsilon_{g}$ as given arguments to $\textrm{OptComp}$ and compute $\delta_{g}$. Recall that the two versions of $\textrm{OptComp}$, viewing $\epsilon_{g}$ as given and computing $\delta_{g}$ and vice versa, are equivalent up to a polynomial factor (just run binary search over values of $\delta_{g}$ computing polynomially many bits of precision). So the formulation we choose for the proof will not affect whether $\textrm{OptComp}$ is in $\#$P or not.  Recall that in our model of computing real valued functions, we will take another input $q$ and we will output an approximation of $\delta_{g}$ to $q$ bits of precision in polynomial time using a $\#P$ oracle where $\delta_{g}$ satisfies the following:
\[
\frac{1}{\prod_{i=1}^{k}{(1+e^{\epsilon_{i}})}}
\sum_{S\subseteq \{1,\ldots,k\}}\max\left\{e^{\sum\limits_{i\in S}\epsilon_{i}} - e^{\epsilon_{g}}\cdot e^{\sum\limits_{i\not\in S}\epsilon_{i}}, 0\right\} = 1-\frac{1-\delta_{g}}{\prod_{i=1}^{k}{(1-\delta_{i})}} 
\]
 Notice that the only part of the expression above that cannot be computed in polynomial time is the summation over subsets of $\{1,\ldots,k\}$. If we knew the sum, computing $\delta_{g}$ would be easy given our inputs. We show how to compute the sum in polynomial time using a $\#$P oracle and it follows that computing $\delta_{g}$ is $\#$P-easy .

Define $f\colon 2^{[k]}\to \mathbb{R}$ as $f(S)= \max\left\{e^{\sum\limits_{i\in S}\epsilon_{i}} - e^{\epsilon_{g}}\cdot e^{\sum\limits_{i\not\in S}\epsilon_{i}}, 0\right\}$. $f$ is computable in polynomial time (to any desired precision). Let $\hat{f}$ be a function computable in polynomial time where $\left|\hat{f}(S)-f(S)\right|<\frac{1}{2^{q+k}}$ for all $S$. Set $m=10^{q}$. Now define the function $g\colon 2^{[k]}\times\mathbb{N} \to \{0,1\}$ as follows:
\[
g(S,n) = \begin{dcases}
1 & \mathrm{if}~m\cdot \hat{f}(S)\geq n \\
0 & \mathrm{otherwise}
\end{dcases}
\]

We can now phrase a decision problem in NP: Does there exist a pair $(S,n)$ such that $g(S,n)=1$? This is in NP because given a witness $(S,n)$, we can compute $m\cdot \hat{f}(S)$ and compare the output to $n$, thereby verifying the solution, in polynomial time. Since this is an NP problem, a $\#P$ oracle can count the number of solutions to it in one time step. Notice that for every set $S$, the number of solutions (pairs of the form $(S,n)$ satisfying $g(S,n)=1$) is exactly $m\cdot \hat{f}(S)$ because $g$ will output $1$ for $g(S,1), g(S,2), \ldots, g(S,m\cdot \hat{f}(S))$. So over all possible sets $S$, the number of solutions as counted by the $\#P$ oracle equals $m\cdot\sum_{S\subseteq [k]}\hat{f}(S)$. Dividing this by $m$ gives us the sum up to an additive error of $\frac{2^{k}}{2^{q+k}}=\frac{1}{2^{q}}$, which can be used to compute $\delta_{g}$ to $q$ bits of precision in polynomial time. This only required one call to a $\#P$ oracle. So computing $\textrm{OptComp}$ is $\#P$-easy.
\end{proof}

Next we show that computing $\textrm{OptComp}$ is also $\#P$-hard through a series of reductions. We start with a multiplicative version of the partition problem that is known to be $\#P$-complete by Ehrgott \cite{Ehr00}. The problems in the chain of reductions are defined below. 

\begin{definition}
$\#\textsc{INT-PARTITION}$ is the following problem: given a set $Z = \{z_{1}, z_{2}, \ldots, z_{k}\}$ of positive integers, count the number of partitions $P\subseteq [k]$ such that
\[\prod_{i\in P}z_{i} - \prod_{i\not\in P}z_{i} = 0\]
\end{definition}

All of the remaining problems in our chain of reductions take inputs $\{w_{1},\ldots,w_{k}\}$ where $1\leq w_{i}\leq e$ is the $Dth$ root of a positive integer $z_{i}$ for all $i\in [k]$ and some positive integer $D$. All of the reductions we present actually hold for every positive integer $D$, including $D=1$ (in which case the inputs are integers). However, we will constrain $D$ to be large enough so that our inputs are in the range $[1,e]$. This is because in the final reduction to $\textrm{OptComp}$, $\epsilon_{i}$ values in the proof are set to $\ln(w_{i})$. We want to show that our reductions hold for reasonable values of $\epsilon$'s in a differential privacy setting so throughout the proofs we use $w_{i}$'s $\in [1,e]$ to correspond to $\epsilon_{i}$'s $\in [0,1]$ in the final reduction. In fact, we will later state our reductions as applying to instances where $\prod_{i}w_{i}\leq e^{\epsilon}$ (and hence $\sum_{i}\epsilon_{i}\leq \epsilon$) for any desired $\epsilon>0$. 

\begin{definition}
$\#\textsc{PARTITION}$ is the following problem: given a number $D\in\mathbb{N}$ and a set $W = \{w_{1}, w_{2}, \ldots, w_{k}\}$ of real numbers where $1\leq w_{1},\ldots,w_{k}\leq e$ are $D$th roots of positive integers $z_{1},\dots z_{k}$, respectively, count the number of partitions $P\subseteq [k]$ such that
\[\prod_{i\in P}w_{i} - \prod_{i\not\in P}w_{i} = 0\]
(The real numbers $w_{1},\ldots, w_{k}$ are specified in the input by $z_{1},\ldots,z_{k}$ and $D$ with the input size being the combined bit length of these integers in binary).
\end{definition}

\begin{definition}
$\#\textsc{T-PARTITION}$ is the following problem: given a number $D\in\mathbb{N}$, a set $W = \{w_{1}, w_{2}, \ldots, w_{k}\}$ of real numbers and a \emph{positive} real number $T$ where $1\leq w_{1},\ldots,w_{k}\leq e$ are $D$th roots of positive integers $z_{1},\dots z_{k}$, respectively, and $T=\sqrt[2D]{t}-\sqrt[2D]{t'}$ for two integers $t,t'$, count the number of partitions $P\subseteq [k]$ such that
\[\prod_{i\in P}w_{i} - \prod_{i\not\in P}w_{i} = T\]
(The real numbers $w_{1},\ldots, w_{k}$ and $T$ are specified in the input by $z_{1},\ldots,z_{k},t,t'$ and $D$ with the input size being the combined bit length of these integers in binary).
\end{definition}

\begin{definition}
$\textsc{SUM-PARTITION}$: given a number $D\in\mathbb{N}$ and a set $W = \{w_{1}, w_{2}, \ldots, w_{k}\}$ of real numbers where $1\leq w_{1},\ldots,w_{k}\leq e$ are $D$th roots of positive integers $z_{1},\dots z_{k}$, respectively, and a rational number $r>1$, find
\[
\sum_{P\subseteq [k]}\max\left\{\prod_{i\in P}w_{i} - r\cdot\prod_{i\not\in P}w_{i},0\right\}
\]
(The real numbers $w_{1},\ldots, w_{k}$ are specified in the input by $z_{1},\ldots,z_{k}$ and $D$ with the input size being the combined bit length of these integers and the numerator and denominator of $r$ in binary).
\end{definition}
Since the output of SUM-PARTITION is irrational, the actual computational problem is defined according to our convention in Section \ref{sect:tech} for computing real-valued functions. That is, given an additional input $q$, compute a number $y$ such that
\[
\left|y-\sum_{P\subseteq [k]}\max\left\{\prod_{i\in P}w_{i} - r\cdot\prod_{i\not\in P}w_{i},0\right\}\right|<\frac{1}{2^{q}}
\]
\noindent We prove that computing $\textrm{OptComp}$ is $\#P$-hard by the following series of reductions:

\[
\#\textsc{INT-PARTITION} \leq \#\textsc{PARTITION} \leq \#\textsc{T-PARTITION} \leq \textsc{SUM-PARTITION} \leq \mathrm{OptComp}
\]

Since $\#\textsc{INT-PARTITION}$ is known to be $\#P$-complete \cite{Ehr00}, the chain of reductions will prove that $\mathrm{OptComp}$ is $\#P$-hard.   

\begin{lemma}
For every constant $c>1$, $\#\textsc{PARTITION}$ is $\#P$-hard, even on instances where $\prod_{i}w_{i}\leq c$.
\end{lemma}
\begin{proof}
Given an instance of $\#\textsc{INT-PARTITION}$, $\{z_{1},\ldots,z_{k}\}$, we show how to find the solution in polynomial time using a $\#\textsc{PARTITION}$ oracle. Set $D=\lceil{\log_{c}(\prod_{i}z_{i})}\rceil$ and $w_{i}=\sqrt[D]{z_{i}} ~\forall i\in [k]$. Note that $\prod_{i}w_{i}=\left(\prod_{i}z_{i}\right)^{1/D}\leq c$.
Let $P\subseteq [k]$:

\begin{align*}
\prod_{i\in P}w_{i} = \prod_{i\not\in P}w_{i}  &\iff \left(\prod_{i\in P}w_{i}\right)^{D} = \left(\prod_{i\not\in P}w_{i}\right)^{D} \\
&\iff \prod_{i\in P}z_{i} = \prod_{i\not\in P}z_{i}
\end{align*}
There is a one-to-one correspondence between solutions to the $\#\textsc{PARTITION}$ problem and solutions to the given $\#\textsc{INT-PARTITION}$ instance. We can solve $\#\textsc{INT-PARTITION}$ in polynomial time with a $\#\textsc{PARTITION}$ oracle. Therefore $\#\textsc{PARTITION}$ is $\#P$-hard.
\end{proof}

\begin{lemma}
For every constant $c>1$, $\#\textsc{T-PARTITION}$ is $\#P$-hard, even on instances where $\prod_{i}w_{i}\leq c$.
\end{lemma}
 
\begin{proof}
Let $c>1$ be a constant. We will reduce from $\#\textsc{PARTITION}$, so consider an instance of the $\#\textsc{PARTITION}$ problem, $W = \{w_{1}, w_{2}, \ldots, w_{k}\}$ of $D$th roots of integers $z_{1},\ldots,z_{k}$, respectively. We may assume $\prod_{i}w_{i}\leq \sqrt{c}$ since $\sqrt{c}$ is also a constant greater than $1$.

Set $W'=W\cup\{w_{k+1}\}$, where $w_{k+1}=\prod_{i=1}^{k}w_{i}$. Notice that $\prod_{i=1}^{k+1}w_{i}\leq (\sqrt{c})^{2}=c$. Set $T = \sqrt{w_{k+1}}\left(w_{k+1}-1\right)$. Notice that $w_{k+1}=\left(\prod_{i=1}^{k}z_{i}\right)^{\frac{1}{D}}$ so by setting integers $t=\left(\prod_{i=1}^{k}z_{i}\right)^{3}$ and $t'=\prod_{i=1}^{k}z_{i}$ we get that 
\[
T=\sqrt[2D]{t}-\sqrt[2D]{t'}
\]
which meets the input requirement for $\#\textsc{T-PARTITION}$. So we can use a $\#\textsc{T-PARTITION}$ oracle to count the number of partitions $Q\subseteq \{1,\ldots,k+1\}$ such that
\[\prod_{i\in Q}w_{i} - \prod_{i\not\in Q}w_{i} = T\]

Let $P=Q\cap\{1,\ldots,k\}$. We will argue that $\prod_{i\in Q}w_{i}-\prod_{i\not\in Q}w_{i}=T$ if and only if $\prod_{i\in P}w_{i}=\prod_{i\not\in P}w_{i}$, which completes the proof. There are two cases to consider: $w_{k+1}\in Q$ and $w_{k+1}\not\in Q$.

$\textbf{Case 1:}$ $w_{k+1}\in Q$. In this case, we have:
\begin{align*}
&w_{k+1}\cdot\left(\prod_{i\in P}w_{i}\right) - \prod_{i\not\in P}w_{i} = \prod_{i\in Q}w_{i}-\prod_{i\not\in Q}w_{i}= T = \sqrt{w_{k+1}}\left(w_{k+1}-1\right) \\
&\iff \left(\prod_{i\in[k]}w_{i}\right)\left(\prod_{i\in P}w_{i}\right)^{2} - \prod_{i\in[k]}w_{i}=\sqrt{\prod_{i\in[k]}w_{i}}\left(\prod_{i\in[k]}w_{i}-1\right)\left(\prod_{i\in P}w_{i}\right)  ~~~\text{ multiplied both sides by} \prod\limits_{i\in P}w_{i}\\
&\iff \left(\prod_{i\in P}w_{i}-\sqrt{\prod_{i\in[k]}w_{i}}\right)\left(\prod_{i\in[k]}w_{i}\prod_{i\in P}w_{i} + \sqrt{\prod_{i\in[k]}w_{i}}\right)=0 ~~~~~~~~~~~~~~~~~~~~~~\text{factored quadratic in} \prod\limits_{i\in P}w_{i} \\
&\iff \prod_{i\in P}w_{i}=\sqrt{\prod_{i\in[k]}w_{i}}\\
&\iff \prod_{i\not\in P}w_{i}=\prod_{i\in P}w_{i}
\end{align*}

So there is a one-to-one correspondence between solutions to the $\#\textsc{T-PARTITION}$ instance $W'$ where $w_{k+1}\in Q$ and solutions to the original $\#\textsc{PARTITION}$ instance $W$.

\vspace*{\baselineskip}
$\textbf{Case 2:}$ $w_{k+1}\not\in Q$. Solutions now look like:

\[
\prod_{i\in P}w_{i} - \prod_{i\in[k]}w_{i}\prod_{i\not\in P}w_{i} = \sqrt{\prod_{i\in[k]}w_{i}}\left(\prod_{i\in[k]}w_{i}-1\right)
\]

One way this can be true is if $w_{i}=1$ for all $i\in[k]$.  We can check ahead of time if our input set $W$ contains all ones. If it does, then there are $2^{k}-2$ partitions that yield equal products (all except $P=[k]$ and $P=\emptyset$) so we can just output $2^{k}-2$ as the solution and not even use our oracle. The only other way to satisfy the above expression is for $\prod_{i\in P}w_{i} > \prod_{i\in [k]}w_{i}$ which cannot happen because $P\subseteq [k]$. So there are no solutions in the case that $w_{k+1}\not\in Q$. 

Therefore the output of the $\#\textsc{T-PARTITION}$ oracle on $W'$ is the solution to the $\#\textsc{PARTITION}$ problem. So $\#\textsc{T-PARTITION}$ is $\#P$-hard.
\end{proof}

For the next two proofs we will make use of the following fact to bound the amount of precision needed when approximating irrational numbers by rational ones in our reductions:

\begin{fact}
\label{mvt}
For all real numbers $y>x$ and functions $f$ that are differentiable on the interval $[x,y]$:
\[
f(y)-f(x)\geq (y-x)\cdot \min_{z\in(x,y)}f'(z)
\]
\end{fact}

\begin{lemma}
For every constant $c>1$, $\textsc{SUM-PARTITION}$ is $\#P$-hard even on instances where $\prod_{i}w_{i}\leq c$ and where there are no partitions $S$ such that $\prod_{i\in S}w_{i}=r\cdot\prod_{i\not\in S}w_{i}$. 
\end{lemma}

\begin{proof}
We will use a $\textsc{SUM-PARTITION}$ oracle to solve $\#\textsc{T-PARTITION}$ given a set $W=\{w_{1}, \ldots, w_{k}\}$ of $D$th roots of positive integers $z_{1},\ldots,z_{k}$, respectively, and a positive real number $T=\sqrt[2D]{t}-\sqrt[2D]{t'}$ for integers $t,t'$ given in the input. Notice that for every $x>0$:
\begin{align*}
\prod_{i\in P}w_{i} - \prod_{i\not\in P}w_{i}=x &\implies \prod_{i\in P}w_{i} - \frac{\prod_{i\in[k]}w_{i}}{\prod_{i\in P}w_{i}} = x \\
&\implies \exists~ j\in\mathbb{Z}^{+} \mathrm{such~that} \sqrt[D]{j}-\frac{\prod_{i\in[k]}w_{i}}{\sqrt[D]{j}}=x
\end{align*}

Above, $j$ must be a positive integer greater than $\left(\prod_{i=1}^{k}z_{i}\right)^{1/2}$, which tells us that the gap in products from every partition must take a particular form. This means that for a given $D$ and $W$, $\#\textsc{x-PARTITION}$ can only be non-zero on a discrete set of possible values of $x$. So given our $\#$T-PARTITION instance we can find a $T'>T$ such that the above has no solutions for $x$ in the interval $(T,T')$. Specifically, solve the above quadratic for $\sqrt[D]{j}$. If $j$ is not an integer, then we know the answer to the $\#$T-PARTITION instance is 0, so assume $j$ is an integer and set $T'=\sqrt[D]{j+1}-\prod_{i}w_{i}/\sqrt[D]{j+1}$. We can also find an interval $(T'',T)$ just below $T$ where no value of $x$ in the interval can yield a solution above by setting $T''=\sqrt[D]{j-1}-\prod_{i}w_{i}/\sqrt[D]{j-1}$. We use these discreteness properties twice in the proof. Also notice that these intervals are not too small:
\begin{claim}
\label{claim1}
$T'-T\geq 2^{-\text{poly}(n)}$ and $T-T''\geq 2^{-\text{poly}(n)}$ where $n$ is the input length (i.e. the bit lengths of the integers $z_{1},\ldots,z_{k}, t, t'$).
\end{claim}
\begin{claim_proof}[Proof of Claim]
\begin{align*}
T'-T&=\sqrt[D]{j+1}-\frac{\prod_{i\in[k]}w_{i}}{\sqrt[D]{j+1}}-\sqrt[D]{j}+\frac{\prod_{i\in[k]}w_{i}}{\sqrt[D]{j}}\\
&\geq \sqrt[D]{j+1}-\sqrt[D]{j}\\
&\geq \frac{1}{D(j+1)}
\end{align*}
where the last inequality follows from Fact \ref{mvt}. This final value is only exponentially small because $j$ is upper bounded by $\prod_{i=1}^{k}z_{i}$, which is at most exponentially large in the bit length of the $z_{i}$'s. A very similar proof shows that $(T'',T)$ is only exponentially small.
\end{claim_proof}

 This means that we can always find $\hat{T}\in (T,T')$ such that $\hat{T}$ is rational and can be fully specified with a bit length that is polynomial in the input length. Fix such a quantity $\hat{T}$. For all $y>0$, define $P^{y} \equiv \{P\subseteq [k] \mid \prod_{i\in P}w_{i} - \prod_{i\not\in P}w_{i} \geq y\}$. Then, since $x$-PARTITION has no solutions for $x\in (T,T')$:
\begin{align*}
\left|\left\{P\subseteq [k] \mid \prod_{i\in P}w_{i} - \prod_{i\not\in P}w_{i} = T\right\}\right| &= \left|P^{T}\backslash P^{\hat{T}}\right| \\
&= \frac{1}{T}\left(\sum_{P\in P^{T}\backslash P^{\hat{T}}}\left(\prod_{i\in P}w_{i} - \prod_{i\not\in P}w_{i}\right)\right)\\
&= \label{8} \frac{1}{T}\left(\sum_{P\in P^{T}}\left(\prod_{i\in P}w_{i} - \prod_{i\not\in P}w_{i}\right) - \sum_{P\in P^{\hat{T}}}\left(\prod_{i\in P}w_{i} - \prod_{i\not\in P}w_{i}\right)\right)
\end{align*}

We now show how to compute the two sums in the final term using the $\textsc{SUM-PARTITION}$ oracle. We will give the procedure for computing $\sum\limits_{P\in P^{T}}\left(\prod\limits_{i\in P}w_{i} - \prod\limits_{i\not\in P}w_{i}\right)$ and the case with $\hat{T}$ will follow by symmetry. The oracle returns a real number, so by our model of computing real valued functions, we will also give the oracle an additional input that specifies the number of bits of precision in its output. Ultimately we only need to approximate each sum to within $\pm T/4$. This will give an approximation to the $\#\textsc{T-PARTITION}$ problem to within $\pm 1/2$, thereby solving it by rounding the approximation because the solution will be an integer. We want to set the input $r$ to the $\textsc{SUM-PARTITION}$ oracle to be $r=r_{T}$ such that for all $P\subseteq [k]$, we have: 

\begin{equation}
\label{r}
\prod\limits_{i\in P}w_{i}-r_{T}\cdot\prod\limits_{i\not\in P}w_{i}\geq 0 \iff \prod\limits_{i\in P}w_{i}-\prod\limits_{i\not\in P}w_{i}\geq T
\end{equation}

Taking $w=\prod_{i\in [k]}w_{i}$ and thinking of $v=\prod_{i\in P}w_{i}$, it suffices that all positive solutions to each of the following two inequalities are the same:

\[
v-r_{T}\frac{w}{v}\geq 0 ~~~\text{and}~~~v-\frac{w}{v}\geq T
\]
The positive solutions to the left one are $v\geq \sqrt{r_{T}w}$, and to the right one are $v\geq (T+\sqrt{T^{2}+4w})/2$. Setting the right-hand sides equal gives 
\begin{equation}
\label{requation}
r_{T}=\frac{\left(T+\sqrt{T^{2}+4w}\right)^{2}}{4w}
\end{equation}

Since $r_{T}$ might be irrational and SUM-PARTITION takes as input rational values of $r$, we need to find a rational $r$ that approximates $r_{T}$ and preserves the set of solutions $P^{T}$. Recall from Claim \ref{claim1} that there is an (only) exponentially small interval $(T'',T)$ below $T$ such that for all $\bar{T}\in (T'',T)$, $P^{T}=P^{\bar{T}}$. This translates to a corresponding interval $(r_{T''},r_{T})$ such that for all $r\in (r_{T''},r_{T})$, equivalence (\ref{r}) holds. Furthermore, this interval is also only exponentially small.

\begin{claim}
\label{claim2}
$r_{T}-r_{T''}\geq 2^{-\text{poly}(n)}$ where $n$ is the input length (i.e. the bit lengths of the integers $z_{1},\ldots,z_{k}, t, t'$).
\end{claim}
\begin{claim_proof}
To see this, view $r_{T}$ from Equation \ref{requation} as a function $r(T)$ of $T$, and calculate the derivative:

\[
r'(T)=\frac{\left(T+\sqrt{T^{2}+4w}\right)^{2}}{2w\cdot \sqrt{T^{2}+4w}},
\]
Fact \ref{mvt} says that:

\begin{align*}
r_{T}-r_{T''}&=r(T)-r(T'') \\
&\geq \left(\min_{z\in(T'',T)}r'(z)\right)\cdot(T-T'')\\
&\geq (T-T'')\cdot\text{poly}(T)
\end{align*}
(Recall that $1\leq w=\prod_{i}w_{i}\leq c$). This is only exponentially small in the input length by Claim \ref{claim1}.
\end{claim_proof}

So we can choose a rational $r\in (r_{T''},r_{T})$ that can be specified with a number of bits that is polynomial in the input length and preserves $P^{T} = \left\{P\subseteq [k]\mid \prod_{i\in P}w_{i}-r\cdot\prod_{i\not\in P}w_{i}\geq 0\right\}$. However the SUM-PARTITION oracle gives us
\[
\sum_{P\subseteq [k]}\max\left\{\prod_{i\in P}w_{i} - r\cdot\prod_{i\not\in P}w_{i},0\right\}=\sum_{P\in P^{T}}\left(\prod_{i\in P}w_{i}-r\cdot\prod_{i\not\in P}w_{i}\right)
\]
whereas we want to compute the right-hand side without the $r$ coefficient. To get this we just pick another rational $r'\in (r_{T''},r_{T})$ such that $r'-r\geq 2^{-\text{poly}(n)}$. If precision were not an issue, we could run our $\textsc{SUM-PARTITION}$ oracle for $r$ and $r'$ and receive the output:

\begin{align*}
&S_{1}=\sum_{P\in P^{T}}\left(\prod_{i\in P}w_{i}-r\cdot\prod_{i\not\in P}w_{i}\right) \\
&S_{2}=\sum_{P\in P^{T}}\left(\prod_{i\in P}w_{i}-r'\cdot\prod_{i\not\in P}w_{i}\right)
\end{align*}

Then the following linear combination of $S_{1}$ and $S_{2}$ gives us what we want:
\[
\frac{r'-1}{r'-r}\cdot S_{1}-\frac{r-1}{r'-r}\cdot S_{2}=\sum_{P\in P^{T}}\left(\prod_{i\in P}w_{i}-\prod_{i\not\in P}w_{i}\right)
\]

\begin{claim}
Computing $S_{1}$ and $S_{2}$ to within $\pm 2^{-\text{poly}(n)}$ yields an approximation of $\sum_{P\in P^{T}}\left(\prod_{i\in P}w_{i}-\prod_{i\not\in P}w_{i}\right)$ to within $\pm T/4$.
\end{claim}
\begin{claim_proof}
We just need to approximate $S_{1}$ and $S_{2}$ to within $\pm \frac{T}{8}\cdot \frac{r'-r}{r'-1}$ to get the desired precision. This additive error is only exponentially small by Claim \ref{claim2}.
\end{claim_proof}

 Running this whole procedure again for $\hat{T}\in (T,T')$, which we fixed above gives us all the information we need to count the number of solutions to the $\#\textsc{T-PARTITION}$ instance we were given. We can solve $\#\textsc{T-PARTITION}$ in polynomial time with four calls to a $\textsc{SUM-PARTITION}$ oracle. Therefore $\textsc{SUM-PARTITION}$ is $\#P$-hard. 
\end{proof}

\noindent Now we prove that computing $\textrm{OptComp}$ is $\#P$-complete. 

\begin{proof}[Proof of Theorem $\ref{thm:hardness}$]
We have already shown that computing $\textrm{OptComp}$ is $\#P$-easy. Here we prove that it is also $\#P$-hard, thereby proving $\#P$-completeness.

We are given an instance $D$, $W=\{w_{1},\ldots,w_{k}\}, r\in\mathbb{Q},$ and $q$ of $\textsc{SUM-PARTITION}$, where $\forall i\in[k]$, $w_{i}$ is the $D$th root of a corresponding integer $z_{i}$, $\prod_{i}w_{i}\leq c$, and $q$ specifies the desired number of bits of precision in the output. If we disregard precision, we would like to set $\epsilon_{i}=\ln(w_{i})~ \forall i\in [k]$, $\delta_{1}=\delta_{2}=\ldots \delta_{k}=0$ and $\epsilon_{g}=\ln(r)$. Note that $\sum_{i}\epsilon_{i}=\ln\left(\prod_{i}w_{i}\right)\leq \ln(c)$. Since we can take $c$ to be an arbitrary constant greater than $1$, we can ensure that $\sum_{i}\epsilon_{i}\leq\epsilon$ for an arbitrary $\epsilon>0$.

Again we will use the version of $\mathrm{OptComp}$ that takes $\epsilon_{g}$ as input and outputs $\delta_{g}$. After using an $\mathrm{OptComp}$ oracle to find $\delta_{g}$ we know the optimal composition equation \ref{eq:1} from Theorem \ref{thm:optcomp} is satisfied:

\[
\frac{1}{\prod_{i=1}^{k}{(1+e^{\epsilon_{i}})}}
\sum_{S\subseteq \{1,\ldots,k\}}\max\left\{e^{\sum\limits_{i\in S}\epsilon_{i}} - e^{\epsilon_{g}}\cdot e^{\sum\limits_{i\not\in S}\epsilon_{i}}, 0\right\} = 1-\frac{1-\delta_{g}}{\prod_{i=1}^{k}{(1-\delta_{i})}}=\delta_{g}
\]

\noindent Thus we can compute:

\begin{align*}
\delta_{g}\cdot\prod_{i=1}^{k}{(1+e^{\epsilon_{i}})}&=\sum_{S\subseteq \{1,\ldots,k\}}\max\left\{e^{\sum\limits_{i\in S}\epsilon_{i}} - e^{\epsilon_{g}}\cdot e^{\sum\limits_{i\not\in S}\epsilon_{i}}, 0\right\}\\
&= \sum_{S\subseteq \{1,\ldots,k\}}\max\left\{\prod_{i\in S}w_{i} - r\cdot \prod_{i\not\in S}w_{i}, 0\right\}
\end{align*}

This last expression is exactly the solution to the instance of $\textsc{SUM-PARTITION}$ we were given. Taking precision into account, the input $\textsc{SUM-PARTITION}$ instance has an additional input $q$ that specifies the desired number of bits of precision in the output and we can only pass OptComp rational values so we will have to approximate $\epsilon_{i}=\ln(w_{i})$ for all $i$ and $\epsilon_{g}=\ln(r)$. Again there is a worry that when we approximate these values the set of partitions $S$ that make $\prod_{i\in S}w_{i}-r\cdot\prod_{i\not\in S}w_{i}>0$ might change. We want to get enough precision in our inputs so that the set of partitions over which we sum does not change and enough precision so that the output is accurate to $q$ bits. We will calculate the approximations required for each of these two goals separately and the final precision that we use will just be the maximum of the two. We prove that we can achieve both of these goals with the next two claims. 

\begin{claim}
\label{sameset}
There exists a polynomial $p(n)$ in the length $n$ of the input (the bit lengths of $z_{1},\ldots,z_{k}, q$, and the numerator and denominator of $r$) such that if $|w_{i}-w'_{i}|\leq 2^{-p(n)}$ for each $i$, then the set of partitions $S$ satisfying 
\[\prod_{i\in S}w_{i}-r\cdot\prod_{i\not\in S}w_{i}>0\]
is the same as the set of partitions satisfying
\[\prod_{i\in S}w'_{i}-r\cdot\prod_{i\not\in S}w'_{i}>0\]
\end{claim}
\begin{claim_proof}
Recall that SUM-PARTITION is $\#$P-hard even on instances where there are no partitions $S$ such that $\prod_{i\in S}w_{i}=r\cdot\prod_{i\not\in S}w_{i}$ so we may assume our input instance of SUM-PARTITION has no such partitions and still prove the hardness of OptComp. So to ensure that we have enough precision such that the set over which we sum does not change, we must make the error smaller than the minimum possible (in absolute value) nonzero outcome of $\prod_{i\in S}w_{i}-r\cdot\prod_{i\not\in S}w_{i}$. We now bound this quantity. Let 
\[
\mathcal{S}=\left\{S\subseteq [k] \mid \prod_{i\in S}w_{i}\not=\prod_{i\not \in S}w_{i}\right\}
\]
Since $r$ is rational, $r=a/b$ for two integers $a$ and $b$. Let $a'=a^{D}$ and $b'=b^{D}$. Then:
\begin{align*}
    \min_{S\in \mathcal{S}}\left\{\left|\prod_{i\in S}w_{i}-r\cdot\prod_{i\not\in S}w_{i}\right|\right\}
     &= \min_{S\in \mathcal{S}}\left\{\left|\left(\prod_{i\in S}z_{i}\right)^{\frac{1}{D}}-\left(\frac{a'}{b'}\prod_{i\not\in S}z_{i}\right)^{\frac{1}{D}}\right|\right\}\\
    &\geq \min_{S\in \mathcal{S}}\left\{\left|\prod_{i\in S}z_{i}-\frac{a'}{b'}\prod_{i\not\in S}z_{i}\right|\cdot\frac{1}{D\left(\prod_{i\in [k]}z_{i}\right)^{(D-1)/D}}\right\}
\end{align*}
Where the last line follows from Fact \ref{mvt} applied to the function $f(x)=x^{1/D}$. $1/\left(\prod_{i\in [k]}z_{i}\right)^{(D-1)/D}$ is only exponentially small because $\prod_{i\in [k]}z_{i}$ is at most exponentially large in the bit length of the integers $z_{1},\ldots,z_{k}$. We claim that $\left|\prod_{i\in S}z_{i}-\frac{a'}{b'}\prod_{i\not\in S}z_{i}\right|$ is at least $1/b'$ for all $S\in\mathcal{S}$. Fix $S\in \mathcal{S}$: 

\begin{align*}
   \left| \prod_{i\in S}z_{i}-\frac{a'}{b'}\prod_{i\not\in S}z_{i}\right|=h &\implies  \left|b'\cdot\prod_{i\in S}z_{i}-a'\cdot\prod_{i\not\in S}z_{i}\right|=h\cdot b'\\
    &\implies h\geq 1/b'
\end{align*}
Where the last implication follows because $b'\cdot\prod_{i\in S}z_{i}-a'\cdot\prod_{i\not\in S}z_{i}$ is just a difference of integers so the closest nonzero value it can take on is $\pm 1$. 
\end{claim_proof}

\begin{claim}
There exists a polynomial $p(n)$ in the length $n$ of the input (the bit lengths of $z_{1},\ldots,z_{k}, q$, and the numerator and denominator of $r$) such that if $|w_{i}-w'_{i}|\leq 2^{-p(n)}$ for each $i$, then
\[
\left|\sum_{S\subseteq \{1,\ldots,k\}}\max\left\{\prod_{i\in S}w'_{i} - r\cdot \prod_{i\not\in S}w'_{i}, 0\right\}-\sum_{S\subseteq \{1,\ldots,k\}}\max\left\{\prod_{i\in S}w_{i} - r\cdot \prod_{i\not\in S}w_{i}, 0\right\}\right| \leq 2^{-q}
\]
\end{claim}
\begin{claim_proof}
 We will choose $p(n)=p_{1}(n)+p_{2}(n)$ where $p_{1}(n)$ is the polynomial that exists from Claim \ref{sameset} and $p_{2}(n)$ will be determined later. Define
\[
S^{+}=\left\{S\subseteq [k] \mid \prod_{i\in S}w_{i}-r\cdot\prod_{i\not\in S}w_{i}>0 \right\}
\]
Claim \ref{sameset} says that:
\[
S^{+}=\left\{S\subseteq [k] \mid \prod_{i\in S}w'_{i}-r\cdot\prod_{i\not\in S}w'_{i}>0 \right\}
\]

Now we can write 
\begin{align*}
   & \left|\sum_{S\subseteq \{1,\ldots,k\}}\max\left\{\prod_{i\in S}w'_{i} - r\cdot \prod_{i\not\in S}w'_{i}, 0\right\}-\sum_{S\subseteq \{1,\ldots,k\}}\max\left\{\prod_{i\in S}w_{i} - r\cdot \prod_{i\not\in S}w_{i}, 0\right\}\right| =\\
   &\left|\sum_{S\in S^{+}}\left(\prod_{i\in S}w'_{i} - r\cdot \prod_{i\not\in S}w'_{i}\right)-\sum_{S\in S^{+}}\left(\prod_{i\in S}w_{i} - r\cdot \prod_{i\not\in S}w_{i}\right)\right| = \\
   &\left|\sum_{S\in S^{+}}\left(\prod_{i\in S}w'_{i}-\prod_{i\in S}w_{i}\right)-\sum_{S\in S^{+}}r\cdot\left(\prod_{i\not\in S}w'_{i}-\prod_{i\not\in S}w_{i}\right)\right|\leq\\
   &\left|\sum_{S\in S^{+}}\left(\prod_{i\in S}w'_{i}-\prod_{i\in S}w_{i}\right)\right|+\left|\sum_{S\in S^{+}}r\cdot\left(\prod_{i\not\in S}w'_{i}-\prod_{i\not\in S}w_{i}\right)\right|
\end{align*}
Bounding each term in the final expression above by $2^{-(q+1)}$ then gives us the accuracy we want. We will show directly how to bound the second term and the argument for the first term follows symmetrically. By hypothesis we have that for all $S\subseteq[k]$:
\begin{align*}
\prod_{i\not\in S}w_{i}'&\leq \prod_{i\not\in S}\left(w_{i}+2^{-p(n)}\right)\\
&\leq \prod_{i\not\in S}\left(1+2^{-p(n)}\right)w_{i} \\
&\leq \left(1+2^{-p(n)}\right)^{k}\cdot\prod_{i\not\in S}w_{i}
\end{align*}
and similarly 
\[
\prod_{i\not\in S}w_{i}'\geq \left(1-2^{-p(n)}\right)^{k}\cdot\prod_{i\not\in S}w_{i}
\]
It follows that for all $S\subseteq[k]$:
\[
\left(\left(1-2^{-p(n)}\right)^{k}-1\right)\cdot \prod_{i\not\in S}w_{i}\leq \left(\prod_{i\not\in S}w'_{i}-\prod_{i\not\in S}w_{i}\right)\leq \left(\left(1+2^{-p(n)}\right)^{k}-1\right)\cdot \prod_{i\not\in S}w_{i}
\]
Since $|S^{+}|\leq 2^{k}$ and $1\leq \prod_{i\not\in S}w_{i}\leq c$ for all $S$ we get: 
\[
2^{k}\cdot r\cdot\left(\left(1-2^{-p(n)}\right)^{k}-1\right)\cdot\leq \sum_{S\in S^{+}}r\cdot\left(\prod_{i\not\in S}w'_{i}-\prod_{i\not\in S}w_{i}\right)\leq 2^{k}\cdot r\cdot\left(\left(1+2^{-p(n)}\right)^{k}-1\right)\cdot c
\]
Picking $p_{2}(n)$ such that $p(n)=p_{1}(n)+p_{2}(n)> 2k+\log (rc)+q+1$ then suffices to bound the absolute value of the sum by $2^{-(q+1)}$. Repeating the same calculation for $\sum_{S\in S^{+}}\left(\prod_{i\in S}w'_{i}-\prod_{i\in S}w_{i}\right)$ will yield the same approximation except without the factor of $r$. So we can bound both terms by $2^{-(q+1)}$ (and therefore their sum by $2^{-q}$) by approximating each $w_{i}$ to a precision that is polynomial in $n$, which proves the claim. 
\end{claim_proof}

 So by the two claims above we can get an approximation of the $\textsc{SUM-PARTITION}$ instance to $q$ bits of precision in polynomial time with access to an $\mathrm{OptComp}$ oracle. Therefore computing $\mathrm{OptComp}$ is $\#P$-hard.
\end{proof}
\end{section}

\begin{section}{Approximation of OptComp}
Although we cannot hope to efficiently compute the optimal composition for a general set of differentially private algorithms (assuming P$\not=$NP or even FP$\not=\#$P), we show in this section that we can approximate $\mathrm{OptComp}$ to arbitrary precision in polynomial time.

\begin{namedtheorem}[Theorem \ref{thm:approx} (restated)]
There is a polynomial-time algorithm that given rational $\epsilon_{1},\ldots,\epsilon_{k}\geq 0, \delta_{1},\ldots \delta_{k}, \delta_{g} \in [0,1),$ and $\eta\in (0,1)$, outputs $\epsilon^{*}$ satisfying 
\[\mathrm{OptComp}((\epsilon_{1},\delta_{1}),\ldots,(\epsilon_{k},\delta_{k}), \delta_{g})\leq \epsilon^{*}\leq\mathrm{OptComp}((\epsilon_{1},\delta_{1}),\ldots,(\epsilon_{k},\delta_{k}), e^{-\eta/2}\cdot\delta_{g})+\eta\]
The algorithm runs in time \[O\left(\frac{k^{3}\cdot\ebar\cdot(1+\ebar)}{\eta}\cdot \log\left(\frac{k^{2}\cdot\ebar\cdot(1+\ebar)}{\eta}\right)\right)\]
 where $\ebar=\sum_{i\in [k]}\epsilon_{i}/k$, assuming constant-time arithmetic operations.
\end{namedtheorem}
We prove Theorem \ref{thm:approx} using the following three lemmas:

\begin{lemma}
\label{lem:pseudo}
Given non-negative integers $a_{1},\ldots,a_{k}$, $B$ and weights $w_{1},\ldots,w_{k}\in\mathbb{Q}$, one can compute 
\[\sum_{\substack {S\subseteq [k]~\mathrm{s.t.}\\
\sum\limits_{i\in S}a_{i}\leq B~~}}\prod_{i\in S}w_{i}\] 
in time $O(Bk)$.
\end{lemma}
Notice that the constraint in Lemma \ref{lem:pseudo} is the same one that characterizes knapsack problems. Indeed, the algorithm we give for computing $\sum_{S\subseteq [k]}\prod_{i\in S}w_{i}$ is a slight modification of the known pseudo-polynomial time algorithm for counting knapsack solutions, which uses dynamic programming. Next we show that we can use this algorithm to approximate $\mathrm{OptComp}$. 

\begin{lemma}
\label{lem:dyer}
Given a rational $e^{\epsilon_{0}}$ with $\epsilon_{0}\geq 0$ and $\epsilon_{1}=a_{1}\cdot\epsilon_{0},\ldots,\epsilon_{k}=a_{k}\cdot\epsilon_{0},\epsilon^{*}=a^{*}\cdot\epsilon_{0}$ for positive integers $a_{1},\ldots,a_{k},a^{*}$ (given as input), and rational $\delta_{1},\ldots \delta_{k}, \delta_{g} \in [0,1)$, there is an algorithm that determines whether or not $\mathrm{OptComp}((\epsilon_{1},\delta_{1}),\ldots,(\epsilon_{k},\delta_{k}),\delta_{g}) \leq \epsilon^{*}$ and runs in time $O\left(k\cdot\sum_{i=1}^{k}a_{i}\right)$ assuming constant-time arithmetic operations.
\end{lemma}

In other words, if the $\epsilon$ values we are given are all integer multiples of some $\epsilon_{0}$ where $e^{\epsilon_{0}}$ is rational, we can determine whether or not the composition of those privacy parameters is $(a^{*}\cdot\epsilon_{0},\delta_{g})$-DP in pseudo-polynomial time, for every positive integer $a^{*}$. Running binary search over integers $a^{*}$, we can find the minimum such integer. When $\epsilon_{0}$ is small, this gives us a good overestimate of the optimal composition of the discrete input privacy parameters. This means that given any inputs $(\epsilon_{1},\delta_{1}),\ldots,(\epsilon_{k},\delta_{k}),\delta_{g}$ to OptComp, we can discretize and polynomially bound the $\epsilon_{i}$ values to new values $\epsilon_{i}'$ for all $i\in[k]$ and use Lemma \ref{lem:dyer} to approximate OptComp$((\epsilon_{1}',\delta_{1}),\ldots,(\epsilon_{k}',\delta_{k}),\delta_{g})$. The next lemma tells us that this is also a good approximation of OptComp$((\epsilon_{1},\delta_{1}),\ldots,(\epsilon_{k},\delta_{k}),\delta_{g})$. 

\begin{lemma}
\label{lem:goodapprox}
For all $\epsilon_{1},\ldots,\epsilon_{k}, c_{1},\ldots, c_{k}\geq0$ and $\delta_{1},\ldots,\delta_{k},\delta_{g}\in [0,1)$:
\[
\mathrm{OptComp}((\epsilon_{1}+c_{1},\delta_{1}),\ldots,(\epsilon_{k}+c_{k},\delta_{k}),\delta_{g})\leq
\mathrm{OptComp}((\epsilon_{1},\delta_{1}),\ldots,(\epsilon_{k},\delta_{k}),e^{-c/2}\cdot\delta_{g})+c
\]
where $c=\sum_{i=1}^{k}c_{i}$
\end{lemma}

\noindent Next we prove the three lemmas and then show that Theorem \ref{thm:approx} follows.

\begin{proof}[Proof of Lemma $\ref{lem:pseudo}$]
We modify Dyer's algorithm for approximately counting solutions to knapsack problems \cite{Dye03}. The algorithm uses dynamic programming. Given non-negative integers $a_{1},\ldots,a_{k}$, $B$, and weights $w_{1},\ldots,w_{k}\in\mathbb{Q}$, define 
\[F(r,s) = \sum_{\substack {S\subseteq [r]~\mathrm{s.t.}\\
\sum\limits_{i\in S}a_{i}\leq s~~}}\prod_{i\in S}w_{i}\]

We want to compute $F(k,B)$. We can find this by tabulating $F(r,s)$ for $(0\leq r\leq k,~ 0\leq s\leq B)$ using the recursion:
\[
F(r,s)=
\begin{dcases}
1 &\mathrm{if}~r=0\\
F(r-1,s) + w_{r}F(r-1, s-a_{r}) &\mathrm{if}~r>0~\mathrm{and}~a_{r}\leq s\\
F(r-1,s) &\mathrm{if}~r>0~\mathrm{and}~~a_{r}>s
\end{dcases}
\]

Each cell $F(r,s)$ in the table can be computed in constant time given earlier cells $F(r',s')$ where $r'<r$. Thus filling the entire table takes time $O(Bk)$. 
\end{proof}

\begin{proof}[Proof of Lemma $\ref{lem:dyer}$]
Given a rational $e^{\epsilon_{0}}\geq 0$ and $\epsilon_{1}=a_{1}\cdot\epsilon_{0},\ldots,\epsilon_{k}=a_{k}\cdot\epsilon_{0}, \epsilon^{*}=a^{*}\cdot\epsilon_{0}$ for positive integers $a_{1},\ldots,a_{k}, a^{*}$ and rational $\delta_{1},\ldots \delta_{k}, \delta_{g} \in [0,1)$ Theorem \ref{thm:optcomp} tells us that answering whether or not
\[\mathrm{OptComp}((\epsilon_{1},\delta_{1}),\ldots,(\epsilon_{k},\delta_{k}),\delta_{g}) \leq \epsilon^{*}\]
is equivalent to answering whether or not the following inequality holds: 

\begin{equation}
\label{q}
\frac{1}{\prod_{i=1}^{k}{(1+e^{\epsilon_{i}})}}
\sum_{S\subseteq \{1,\ldots,k\}}\max\left\{e^{\sum\limits_{i\in S}\epsilon_{i}} - e^{\epsilon^{*}}\cdot e^{\sum\limits_{i\not\in S}\epsilon_{i}}, 0\right\} \leq 1-\frac{1-\delta_{g}}{\prod_{i=1}^{k}{(1-\delta_{i})}} 
\end{equation}

The right-hand side and $\prod_{i=1}^{k}(1+e^{\epsilon_{i}})$ are easy to compute given the inputs (note that $e^{\epsilon_{i}}$ is rational for all $i\in [k]$ because each is an integer power of $e^{\epsilon_{0}}$). So in order to check the inequality, we will show how to compute the sum. Define 

\begin{align*}
K&=\left\{T\subseteq[k] \mid \sum\limits_{i\not\in T}\epsilon_{i}\geq \epsilon^{*}+\sum\limits_{i\in T}\epsilon_{i}\right\}\\
&=\left\{T\subseteq[k] \mid \sum\limits_{i\in T}\epsilon_{i}\leq \left(\sum\limits_{i=1}^{k}\epsilon_{i}-\epsilon^{*}\right)/2\right\}\\
&=\left\{T\subseteq[k] \mid \sum\limits_{i\in T}a_{i}\leq B\right\}~\mathrm{for}~B=\left\lfloor{\left(\sum\limits_{i=1}^{k}a_{i}-a^{*}\right)/2}\right\rfloor
\end{align*}
and observe that by setting $T=S^{\mathsf{c}}$, we have
\[
\sum_{S\subseteq \{1,\ldots,k\}}\max\left\{e^{\sum\limits_{i\in S}\epsilon_{i}} - e^{\epsilon^{*}}\cdot e^{\sum\limits_{i\not\in S}\epsilon_{i}}, 0\right\} = \sum_{T\in K}\left(\left(\prod_{i=1}^{k}e^{\epsilon_{i}}\cdot\prod\limits_{i\in T}e^{-\epsilon_{i}}\right)-\left(e^{\epsilon^{*}}\cdot\prod\limits_{i\in T}e^{\epsilon_{i}}\right)\right) 
\]

We can now use Lemma \ref{lem:pseudo} to compute each term separately since $K$ is a set of knapsack solutions. Specifically, setting $w_{i}=e^{-\epsilon_{i}}~\forall i\in [k]$, Lemma \ref{lem:pseudo} tells us that we can compute $\sum_{T\subseteq [k]}\prod_{i\in T}w_{i}$ subject to $\sum_{i\in T}a_{i}\leq B$, which is equivalent to $\sum_{T\in K}\prod_{i\in T}e^{-\epsilon_{i}}$. To compute $\sum_{T\in K}\prod_{i\in T}e^{\epsilon_{i}}$, we instead set $w_{i}=e^{\epsilon_{i}}$ and run the same procedure. (Note that $e^{\epsilon^{*}}=(e^{\epsilon_{0}})^{a^{*}}$, which is rational.) So we can determine whether or not Inequality \ref{q} holds. We used the algorithm from Lemma \ref{lem:pseudo} so the running time is $O(Bk)=O\left(k\cdot\sum_{i=1}^{k}a_{i}\right)$ 
\end{proof}

\begin{proof}[Proof of Lemma $\ref{lem:goodapprox}$]
Fix $\epsilon_{1},\ldots,\epsilon_{k}, c_{1},\ldots, c_{k}\geq0$ and $\delta_{1},\ldots,\delta_{k},\delta_{g}\in [0,1)$ and let $c=\sum_{i\in [k]}c_{i}$. Let $\mathrm{OptComp}((\epsilon_{1},\delta_{1}),\ldots,(\epsilon_{k},\delta_{k}),e^{-c/2}\cdot\delta_{g})=\epsilon_{g}$. From Equation \ref{eq:1} in Theorem \ref{thm:optcomp} we know:
\[
\frac{1}{\prod_{i=1}^{k}{(1+e^{\epsilon_{i}})}}
\sum_{S\subseteq \{1,\ldots,k\}}\max\left\{e^{\sum\limits_{i\in S}\epsilon_{i}} - e^{\epsilon_{g}}\cdot e^{\sum\limits_{i\not\in S}\epsilon_{i}}, 0\right\} \leq 1-\frac{1-e^{-c/2}\cdot\delta_{g}}{\prod_{i=1}^{k}{(1-\delta_{i})}} 
\]
Multiplying both sides by $e^{c/2}$ gives:
\begin{align*}
\frac{e^{c/2}}{\prod_{i=1}^{k}{(1+e^{\epsilon_{i}})}}
\sum_{S\subseteq \{1,\ldots,k\}}\max\left\{e^{\sum\limits_{i\in S}\epsilon_{i}} - e^{\epsilon_{g}}\cdot e^{\sum\limits_{i\not\in S}\epsilon_{i}}, 0\right\} &\leq e^{c/2}\cdot\left(1-\frac{1-e^{-c/2}\cdot\delta_{g}}{\prod_{i=1}^{k}{(1-\delta_{i})}}\right) \\
&\leq 1-\frac{1-\delta_{g}}{\prod_{i=1}^{k}{(1-\delta_{i})}}
\end{align*}
The above inequality together with Theorem \ref{thm:optcomp} means that showing the following will complete the proof: 
\[
\sum_{S\subseteq \{1,\ldots,k\}}\max\left\{e^{\sum\limits_{i\in S}(\epsilon_{i}+c_{i})} - e^{\epsilon_{g}+c}\cdot e^{\sum\limits_{i\not\in S}(\epsilon_{i}+c_{i})}, 0\right\}\leq\frac{e^{c/2}\cdot\prod_{i=1}^{k}{(1+e^{\epsilon_{i}+c_{i}})}}{\prod_{i=1}^{k}{(1+e^{\epsilon_{i}})}}
\sum_{S\subseteq \{1,\ldots,k\}}\max\left\{e^{\sum\limits_{i\in S}\epsilon_{i}} - e^{\epsilon_{g}}\cdot e^{\sum\limits_{i\not\in S}\epsilon_{i}}, 0\right\}
\]
Since $(1+e^{\epsilon_{i}+c_{i}})/(1+e^{\epsilon_{i}})\geq e^{c_{i}/2}$ for every $\epsilon_{i}, c_{i}>0$, it suffices to show:
\[
\sum_{S\subseteq \{1,\ldots,k\}}\max\left\{e^{\sum\limits_{i\in S}(\epsilon_{i}+c_{i})} - e^{\epsilon_{g}+c}\cdot e^{\sum\limits_{i\not\in S}(\epsilon_{i}+c_{i})}, 0\right\}\leq \sum_{S\subseteq \{1,\ldots,k\}}e^{c}\cdot\max\left\{e^{\sum\limits_{i\in S}\epsilon_{i}} - e^{\epsilon_{g}}\cdot e^{\sum\limits_{i\not\in S}\epsilon_{i}}, 0\right\}
\]
This inequality holds term by term. If a right-hand term is zero $\left(\sum_{i\in S}\epsilon_{i}\leq\epsilon_{g}+\sum_{i\not\in S}\epsilon_{i}\right)$, then so is the corresponding left-hand term $\left(\sum_{i\in S}(\epsilon_{i}+c_{i})\leq\epsilon_{g}+c+\sum_{i\not\in S}(\epsilon_{i}+c_{i})\right)$. For the nonzero terms, the factor of $e^{c}$ ensures that the right-hand terms are larger than the left-hand terms. 
\end{proof}
\end{section}

\begin{proof}[Proof of Theorem $\ref{thm:approx}$]
Lemma \ref{lem:dyer} tells us that we can determine whether a set of privacy parameters satisfies some $(\epsilon_{g},\delta_{g})$ differential privacy guarantee if the $\epsilon_{i}$ values and $\epsilon_{g}$ are all positive integer multiples of some $\epsilon_{0}$ where $e^{\epsilon_{0}}$ is rational. We are given rational $\epsilon_{1},\ldots,\epsilon_{k}\geq 0, \delta_{1},\ldots \delta_{k}, \delta_{g} \in [0,1),$ and $\eta\in (0,1)$. Let $\ebar=\sum_{i\in[k]}\epsilon_{i}/k$ be the arithmetic mean of the $\epsilon_{i}$ values. Let $\beta=\eta/(k\cdot(1+\ebar)+1)$, set $\epsilon_{0}=\ln(1+\beta)$, and for all $i\in [k]$ set $a_{i}=\lceil{\epsilon_{i}\cdot(1/\beta+1)}\rceil$ and $\epsilon_{i}'=\epsilon_{0}\cdot a_{i}$. We will use the following bounds on $\epsilon_{0}$ in the proof:
\[
\frac{\beta}{2}\leq\frac{\beta}{1+\beta}\leq \epsilon_{0}\leq \beta
\]

With these settings, the $a_{i}$'s are non-negative integers, the $\epsilon_{i}'$ values are all integer multiples of $\epsilon_{0}$ and $e^{\epsilon_{0}}$ is rational. So for every positive integer $a$ we can apply Lemma \ref{lem:dyer} to determine whether or not OptComp$((\epsilon_{1}',\delta_{1}),\ldots,(\epsilon_{k}',\delta_{k}),\delta_{g})\leq a\cdot \epsilon_{0}$ in time $O\left(k\cdot \sum_{i\in [k]}a_{i}\right)$. Running binary search over integers $a$, we can find the minimum such integer, which we will call $a^{*}$. The algorithm's estimate of OptComp$((\epsilon_{1},\delta_{1}),\ldots,(\epsilon_{k},\delta_{k}),\delta_{g})$ will be $a^{*}\cdot\epsilon_{0}$. However since this number is irrational, we will use the Taylor approximation of the natural logarithm to output $\epsilon^{*}$ satisfying $a^{*}\cdot\epsilon_{0}\leq \epsilon^{*}\leq a^{*}\cdot\epsilon_{0}+\beta-\epsilon_{0}$. Since we only need to calculate a few terms of the Taylor expansion of $\ln(1+\beta)$ to achieve this approximation, this step will not affect our running time. 

Since we choose $a^{*}$ to be the minimum integer satisfying composition we have:
\[\epsilon^{*}-\beta\leq(a^{*}-1)\cdot\epsilon_{0}\leq \mathrm{OptComp}((\epsilon_{1}',\delta_{1}),\ldots,(\epsilon_{k}',\delta_{k}),\delta_{g}) \leq a^{*}\cdot\epsilon_{0}\leq\epsilon^{*}\]

$a^{*}$ can range from $0$ to $\sum_{i\in [k]}a_{i}$ so the binary search can be done in $\log\left(\sum_{i\in[k]}a_{i}\right)=\log O\left(k^{2}\cdot\ebar\cdot(1+\ebar)/\eta\right)$ iterations. This gives us a total running time of:
\[O\left(\frac{k^{3}\cdot\ebar\cdot(1+\ebar)}{\eta}\cdot \log\left(\frac{k^{2}\cdot\ebar\cdot(1+\ebar)}{\eta}\right)\right)\]

Now we argue that $\epsilon^{*}$ is a good approximation of OptComp$((\epsilon_{1},\delta_{1}),\ldots,(\epsilon_{k},\delta_{k}),\delta_{g})$. For all $i\in [k]$ we have:
\begin{align*}
   \epsilon'_{i}&=\epsilon_{0}\cdot a_{i}\\
   &\geq \frac{\beta}{1+\beta}\cdot\left\lceil{\epsilon_{i}\cdot\left(\frac{1}{\beta}+1\right)}\right\rceil\\
   &\geq \epsilon_{i}
\end{align*}
So all of the $\epsilon_{i}'$ values are overestimates of their corresponding $\epsilon_{i}$ values and therefore
\[
\mathrm{OptComp}((\epsilon_{1},\delta_{1}),\ldots,(\epsilon_{k},\delta_{k}),\delta_{g})\leq \mathrm{OptComp}((\epsilon_{1}',\delta_{1}),\ldots,(\epsilon_{k}',\delta_{k}),\delta_{g}) \leq \epsilon^{*}
\] 
satisfying one of the inequalities in the theorem. We also have for all $i\in [k]$:
\begin{align*}
    \epsilon'_{i} &=\epsilon_{0}\cdot\left\lceil{\epsilon_{i}\cdot\left(\frac{1}{\beta}+1\right)}\right\rceil\\
    &\leq \beta\cdot\left(\epsilon_{i}\cdot\left(\frac{1}{\beta}+1\right)+1\right)\\
    &=\epsilon_{i}+\beta\cdot(\epsilon_{i}+1)\\
\end{align*}
Let $c_{i}=\beta\cdot(\epsilon_{i}+1)$ for all $i\in [k]$ and let $c=\sum_{i\in[k]}c_{i}=\beta\cdot k\cdot(1+\ebar)$. Now we get
\begin{align*}
\epsilon^{*}-\beta&\leq \mathrm{OptComp}((\epsilon_{1}',\delta_{1}),\ldots,(\epsilon_{k}',\delta_{k}),\delta_{g})\\
&\leq \mathrm{OptComp}((\epsilon_{1}+c_{1},\delta_{1}),\ldots,(\epsilon_{k}+c_{k},\delta_{k}),\delta_{g})\\
&\leq \mathrm{OptComp}((\epsilon_{1},\delta_{1}),\ldots,(\epsilon_{k},\delta_{k}),e^{-\beta\cdot k\cdot(1+\ebar)/2}\cdot\delta_{g})+\beta\cdot k\cdot(1+\ebar)
\end{align*}
by Lemma \ref{lem:goodapprox}. Noting that $\beta\cdot k\cdot(1+\ebar)$ and $\beta\cdot k\cdot(1+\ebar)+\beta$ are both at most $\eta$ completes the proof.
\end{proof}

\newpage

\appendix
\section{Comparison of Composition Theorems} \label{app:A}
The figures below compare the performances of four homogeneous composition theorems. In all figures, ``Summing'' refers to basic composition - Theorem \ref{thm:basiccomp} \cite{DKMMN06}, ``DRV'' refers to advanced composition - Theorem \ref{thm:advancedcomp} \cite{DRV10}, ``KOV Bound'' refers to a bound in \cite{KOV15} that is a closed form approximation of the optimal composition theorem, and ``Optimal'' refers to the optimal composition theorem - Theorem \ref{thm:homogeneouscomp} \cite{KOV15}. Here we are composing $k$ mechanisms that are $(\epsilon,\delta)$ differentially private to obtain an $(\epsilon_{g},\delta_{g})$ differentially private mechanism as guaranteed by one of the composition theorems. 

\captionsetup[subfigure]{labelformat=empty}
\begin{figure}[h]
\begin{subfigure}{0.5\textwidth}
\includegraphics[scale=.45]{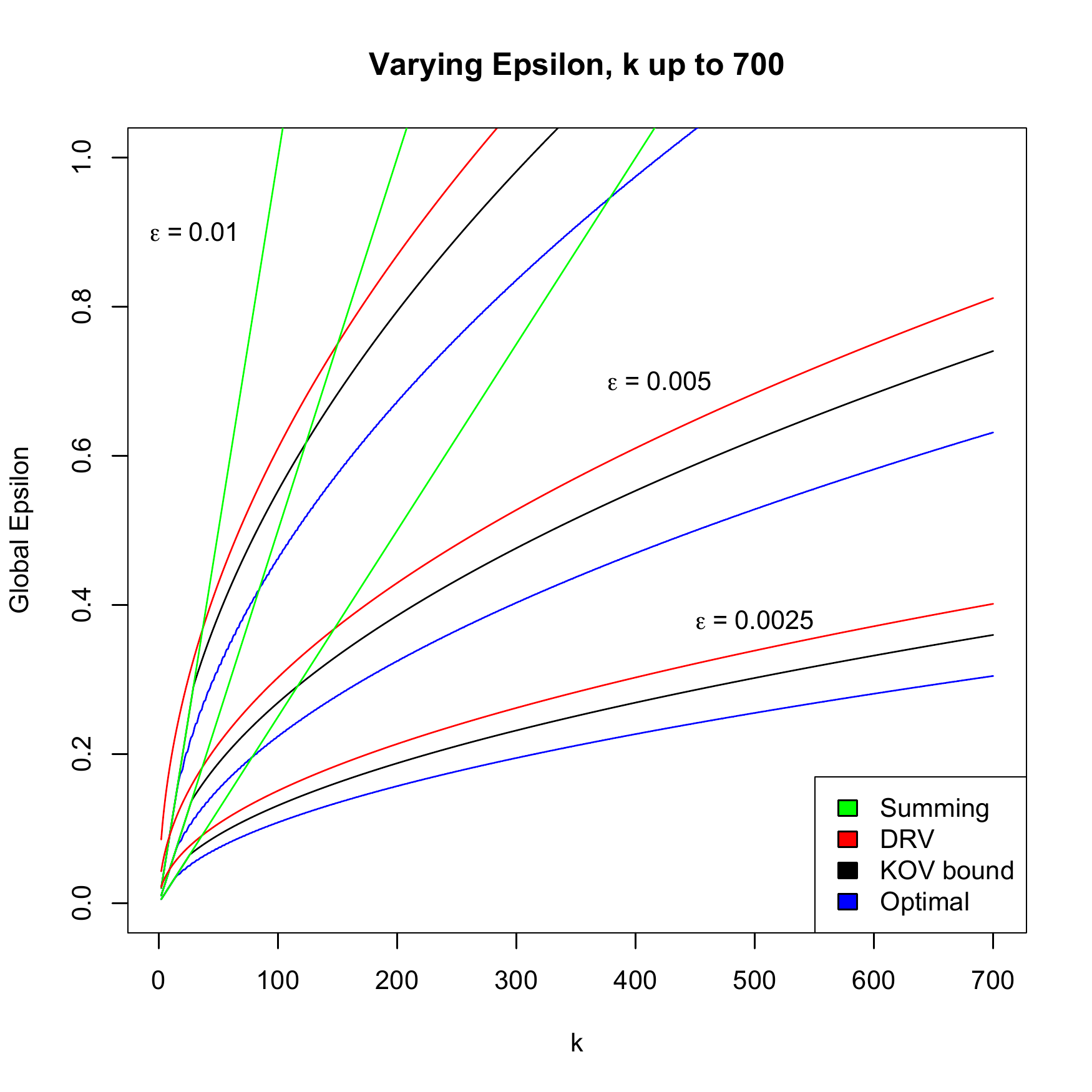} 
\label{fig:subim1}
\end{subfigure}
\begin{subfigure}{0.5\textwidth}
\includegraphics[scale=.45]{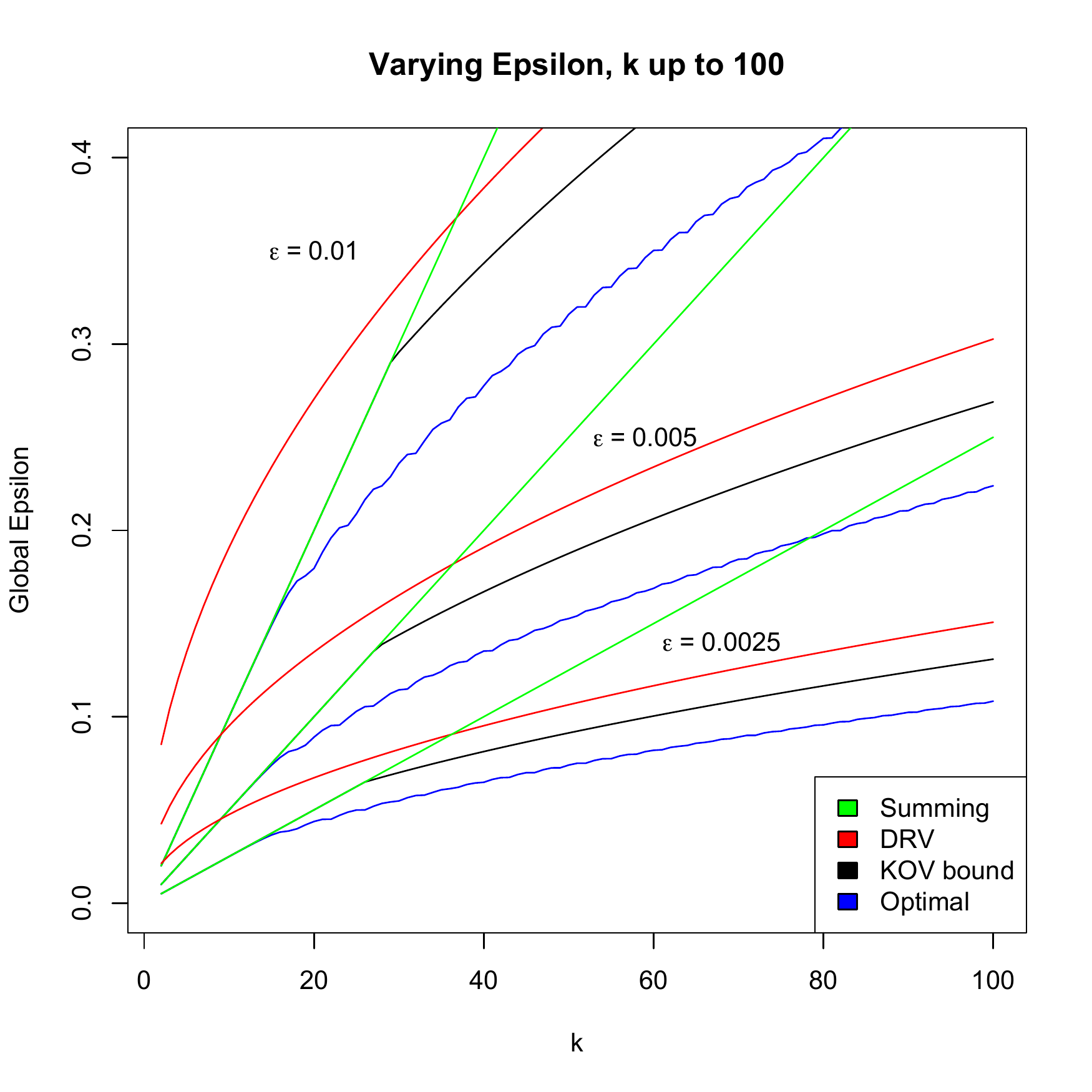}
\label{fig:subim2}
\end{subfigure}
\caption{(Left) $\epsilon_{g}$ given by four composition theorems for varying values of $\epsilon$ as $k$ grows. Parameters $\delta=0$ and $\delta_{g}=2^{-25}$. (Right) Same plot zoomed in on the $k<100$ regime. We see that optimal composition gives substantial savings in $\epsilon_{g}$, even for moderate values of $k$.}
\label{fig:image2}
\end{figure}

\begin{figure}[t]
\begin{subfigure}{0.5\textwidth}
\includegraphics[scale=.45]{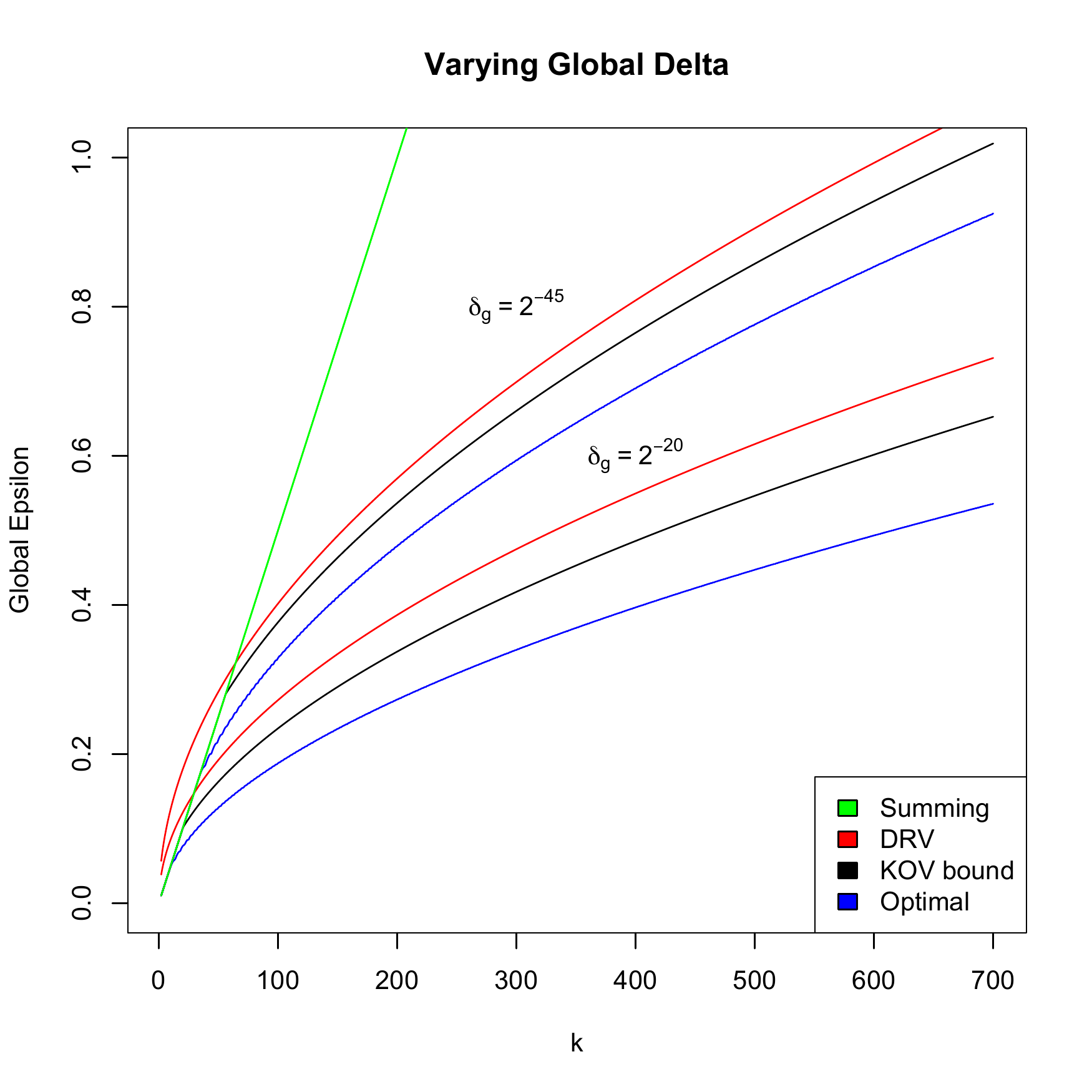} 
\label{fig:subim1}
\end{subfigure}
\begin{subfigure}{0.5\textwidth}
\includegraphics[scale=.45]{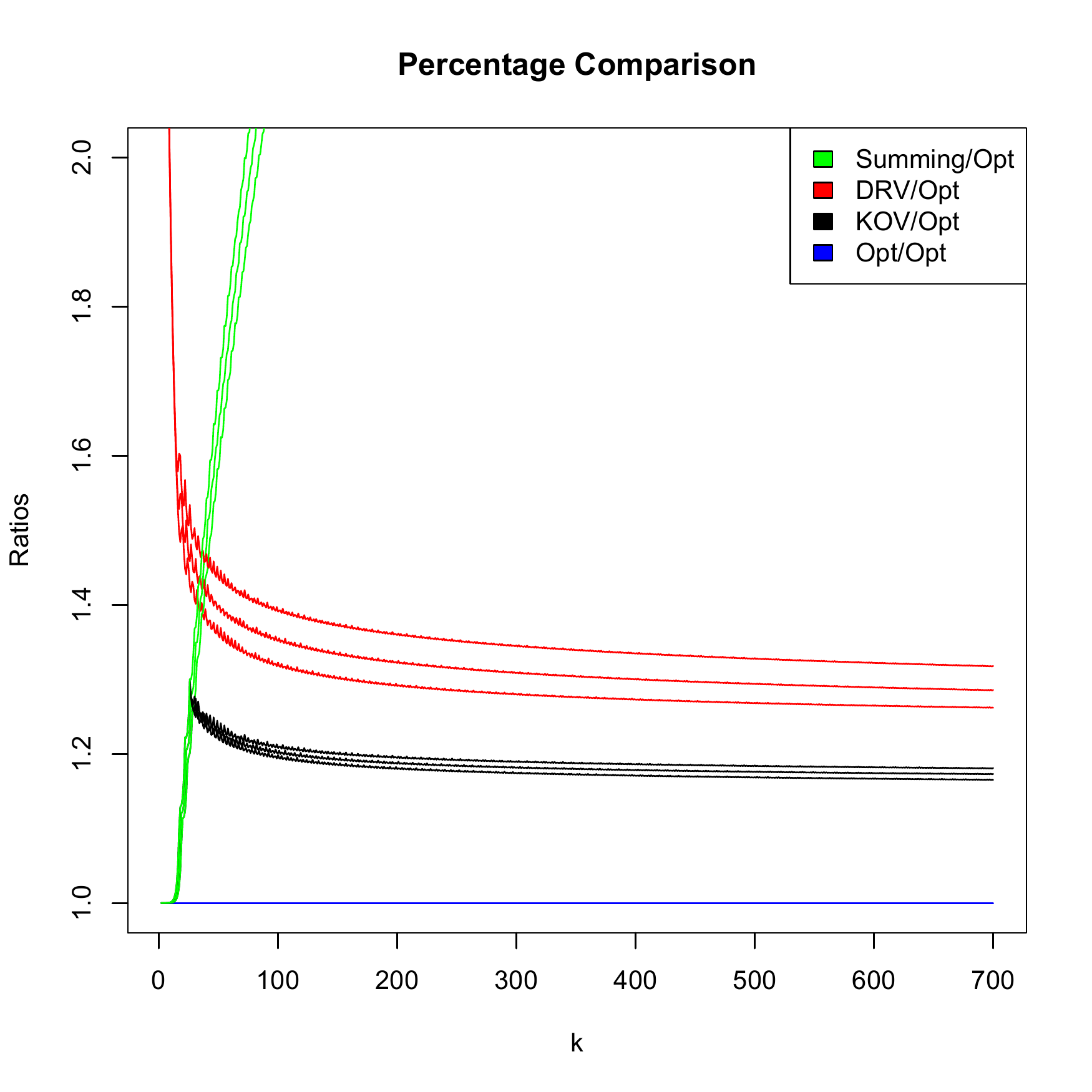}
\label{fig:subim2}
\end{subfigure}
\caption{(Left) $\epsilon_{g}$ given by four composition theorems for varying values of $\delta_{g}$ as $k$ grows, with parameters $\delta=0$ and $\epsilon=.005$ for the individual mechanisms. $\delta_{g}$ does not affect $\epsilon_{g}$ in basic composition. (Right) Performance of composition theorems measured relative to optimal composition. Depicts every curve in Figure 1 divided by the optimal composition curve. We see that relative performances of the KOV bound and DRV seem to converge to a constant. The $\epsilon_{g}$ values given by the KOV bound are about $20\%$ larger than optimal and the values given by advanced composition are about 30-40$\%$ larger than optimal.}
\label{fig:image2}
\end{figure}
\end{document}